\documentclass[11pt]{article} 

\usepackage[top=1in, bottom=1in, left=1in, right=1in]{geometry}

\usepackage{charter}

\usepackage{lipsum}
\usepackage{comment,xspace}
\usepackage[dvipsnames]{xcolor} 
\definecolor{ptblue}{RGB}{15,76,129} 
\definecolor{ptemerald}{HTML}{009473} 
\definecolor{bluegray}{rgb}{0.4, 0.6, 0.8}
\definecolor{ptilluminating}{HTML}{F5DF4D} 
\definecolor{ptgray}{HTML}{939597} 
\usepackage{cprotect}
\usepackage[normalem]{ulem} 

\usepackage[ruled,linesnumbered,vlined]{algorithm2e}
\usepackage{algorithmic}

\SetCommentSty{mycommfont}

\usepackage{comment,xspace,enumitem}
\usepackage{booktabs}
\usepackage{wrapfig}
\usepackage{centernot}

\usepackage{changepage}
\usepackage{booktabs,multirow,subcaption}

\usepackage[breakable, theorems, skins]{tcolorbox}

\usepackage{amsmath,amsfonts,amssymb,amsthm,mathtools,thmtools} 
\DeclareMathOperator*{\argmax}{arg\,max}
\DeclareMathOperator*{\argmin}{arg\,min}

\usepackage[square,numbers,sort]{natbib} 

\usepackage{hyperref}
\hypersetup{
linktocpage,
colorlinks=true,
citecolor=cobalt, 
urlcolor=ptblue, 
linkcolor=cobalt, 
}
\usepackage{cleveref} 

\usepackage{enumitem}
\usepackage{tikz} 
\theoremstyle{plain}
\newtheorem{theorem}{Theorem}[section]
\newtheorem{corollary}[theorem]{Corollary}
\newtheorem{proposition}[theorem]{Proposition}
\newtheorem{lemma}[theorem]{Lemma}

\newtheorem{claim}[theorem]{Claim}

\theoremstyle{definition}
\newtheorem{definition}[theorem]{Definition}

\newtheorem*{theorem*}{Theorem}

\theoremstyle{remark}

\definecolor{cobalt}{rgb}{0.0, 0.28, 0.67}

\newcommand{\A}{\mathcal{A}}
\newcommand{\B}{\mathcal{B}}
\newcommand{\C}{\mathcal{C}}

\newcommand{\calH}{\mathcal{H}}
\newcommand{\calX}{\mathcal{X}}

\newcommand{\calA}{\mathcal{A}}
\newcommand{\calB}{\mathcal{B}}
\newcommand{\SW}{\mathrm{SW}}
\newcommand{\Opt}{\mathcal{O}}
\newcommand{\Rn}{\mathbb{R}_{\ge 0}^n}

\newcommand{\supp}{\mathsf{supp}}
\newcommand{\Pin}[1]{\Pi_{n}(#1)}

\newcommand{\EF}[1]{\mathrm{EF#1}}
\newcommand{\EFR}[1]{\mathrm{EFR}\text{-}#1}
\newcommand{\EFRn}{\mathrm{EFR}\text{-}(n-1)}
\newcommand{\PO}{\mathrm{PO}}

\title{\bfseries Fair and Efficient Allocation of Indivisible Mixed Manna}

 \author{Siddharth Barman\thanks{Indian Institute of Science, Bangalore;  {barman@iisc.ac.in} }  \and Vishwa Prakash HV\thanks{Chennai Mathematical Institute, Chennai;  {vishwa@cmi.ac.in} }  \and Aditi Sethia\thanks{Indian Institute of Science, Bangalore;  {aditisethia@iisc.ac.in} } \and  Mashbat Suzuki\thanks{UNSW Sydney; {mashbat.suzuki@unsw.edu.au}}}

\date{}

\begin{document}

\maketitle

\begin{abstract}
We study fair division of indivisible mixed manna (items whose values may be positive, negative, or zero) among agents with additive valuations. Here, we establish that fairness---in terms of a relaxation of envy-freeness---and Pareto efficiency can always be achieved together. Specifically, our fairness guarantees are in terms of envy-freeness up to $k$ reallocations ($\mathrm{EFR}$-${k}$): An allocation $\mathcal{A}$ of the indivisible items is said to be $\mathrm{EFR}$-${k}$ if there exists a subset $R$ of at most $k$ items such that, for each agent $i$, we can reassign items from within $R$ (in $\mathcal{A}$) and obtain an allocation, $\mathcal{A}^i$, which is envy-free for $i$. We establish that, when allocating mixed manna among $n$ agents with additive valuations, an $\mathrm{EFR}$-$(n-1)$ and Pareto optimal (PO) allocation $\mathcal{A}$ always exists. Further, the individual envy-free allocations $\mathcal{A}^i$, induced by reassignments, are also PO. In addition, we prove that such fair and efficient allocations are efficiently computable when the number of agents, $n$, is fixed. 

We also obtain positive results focusing on $\mathrm{EFR}$ by itself (and without the PO desideratum). Specifically, we show that an $\mathrm{EFR}$-$(n-1)$ allocation of mixed manna can be computed in polynomial time. In addition, we prove that when all the items are goods, an $\mathrm{EFR}$-${\lfloor n/2 \rfloor}$ allocation exists and can be computed efficiently. Here, the $(n-1)$ bound is tight for chores and $\lfloor n/2 \rfloor$ is tight for goods. 

Our results advance the understanding of fair and efficient allocation of indivisible mixed manna and rely on a novel application of the Knaster-Kuratowski-Mazurkiewicz (KKM) Theorem in discrete fair division. We utilize weighted welfare maximization, with perturbed valuations, to achieve Pareto efficiency, and overall, our techniques are notably different from existing market-based approaches. 
\end{abstract}

\section{Introduction}
\label{sec:introduction}

Fair division of resources and chores is a frequently encountered problem in many real-world settings, including estate settlements, task assignments, border disputes, and distribution of liabilities. Questions of who gets what, specifically when the participating agents have individual preferences, are both practically pressing and theoretically rich. Motivated by such considerations and over the past several decades, a significant body of work in mathematics, economics, and computer science has been developed to formally address fair division \cite{bramstaylorFairDivision96, moulinHandbookComputational16, moulinFairDivision03}. These works span axiomatic foundations, existential guarantees, and algorithmic results for allocating divisible and indivisible items. The items can be goods (items with nonnegative values), chores (negatively valued items), or mixed manna (items whose value may vary in sign across the agents). 

In this literature, two recurring desiderata are fairness and economic efficiency. Among the multiple notions of fairness, envy-freeness \cite{foleyResourceallocation67} stands as a quintessential one; this fairness criterion requires that no agent prefers another agent's bundle to its own. Further, a standard notion of efficiency here is Pareto optimality ($\PO$): no other allocation should make some agent better off without making someone else worse off. These two notions have been extensively studied both in isolation and in conjunction -- they constitute canonical standards for evaluating allocations and methods in fair division. 

Much of the classic work in fair division has focused on divisible goods --- resources that can be fractionally allocated, e.g., land and time. In such settings, prominent results by Varian \cite{varianEquityenvy74} and Weller \cite{wellerFairdivision85} show that envy-freeness and $\PO$ can be achieved together; see also \cite{barbanel2005geometry}

Since an envy-free allocation of indivisible items (i.e., items that cannot be fractionally assigned) is not guaranteed to exist, applicable relaxations of envy-freeness have been formulated in discrete fair division. A prominent relaxation, {envy-freeness up to one item} (\(\EF{1}\)), requires that existing envy can be eliminated by the removal of a single item from the envied bundle. \(\EF{1}\) was introduced by Budish~\cite{budishCombinatorialAssignment11} (see also \cite{LMM04}) and has been extensively studied in recent years.  

For the case of indivisible goods (items with nonnegative values) and under additive valuations, positive results are known for simultaneously achieving \(\EF{1}\) and \(\PO\). Caragiannis et al.~\cite{CKM+19unreasonable} showed that allocations maximizing the Nash social welfare satisfy both properties. Barman et al.~\cite{BKV2018} developed a pseudo-polynomial time algorithm to compute such allocations using a Fisher market framework. In addition, the work of  Mahara~\cite{mahara2024} shows that when the number of agents is fixed, one can efficiently compute an $\EF{1}$ and $\PO$ allocation of goods under additive valuations.  

In notable contrast to the goods setting, the existence of $\EF{1}$ and $\PO$  allocations for indivisible chores---and moreover for mixed manna---remains largely unresolved. For instance, even for the case of four agents with additive valuations, it is not known whether an \(\EF{1}\) and \(\PO\) allocation of chores always exists. For mixed manna, this problem is, in fact, open for three agents. The following observations further highlight the dichotomy: There does not exist a welfare function (analogous to Nash social welfare) that guarantees \(\EF{1}\) for chores \cite{eckart2024fairness}. In addition, the termination of the market-based approaches (which guarantee $\EF{1}$ and $\PO$ for goods) do not carry forward directly when the valuations are negative. 

To make progress on the indivisible chores front, prior works have focused on specialized settings. For instance, the existence of $\EF{1}$ and $\PO$ allocations is known for chores among three agents~\cite{gargetalNewAlgorithms23}. Further, for chores, fairness guarantees under further relaxations---such as $(n-1)$ multiplicative approximations of $\EF{1}$---have also been obtained \cite{GargConstantFactorSTOC25}.   

The mixed manna setting generalizes chores division; indeed, the mixed manna setup captures the allocation of goods and chores together, with valuations that are neither monotone increasing nor monotone decreasing (i.e., with non-monotone valuations). Under mixed manna, the existence of $\EF{1}$ and $\PO$ allocations is known specifically for two agents \cite{azizetalFairallocation22} and under identical or ternary valuations \cite{aleksandrovwalshTwoAlgorithms20}. It is relevant to note that, prior to the current work, in the general setting of arbitrarily many agents and additive valuations, fairness and efficiency guarantees for mixed manna were not known under any envy-based fairness notion.

This work contributes to this thread of research in discrete fair division on fairness in conjunction with efficiency. We establish that, under a relaxation of envy-freeness, fair and efficient allocations of mixed manna always exist among agents with additive valuations. Notably, our existential guarantee is not confined to settings with a fixed number of agents or to a specific subclass of additive valuations. 

The current paper studies the fairness notion of envy-freeness up to $k$ reallocations ($\EFR{k}$). Specifically, an allocation $\calA$ of the indivisible items is said to be $\EFR{k}$ if there exists a (fixed) subset $R$ of at most $k$ items such that, for each agent $i$, we can reassign items from within $R$ (in $\calA$) and obtain an allocation, $\calA^i$, which is envy-free for $i$. We establish that this fairness notion is compatible with Pareto optimality: When allocating mixed manna among $n \in \mathbb{Z}_+$ agents with additive valuations, an $\EFR{(n-1)}$ and $\PO$ allocation $\calA$ always exists. 

In fact, in this existential guarantee the individual envy-free allocations (obtained by reassigning items from within $R$ in $\calA$) are all Pareto optimal. That is, we prove that there exist proximal allocations $\calA^1,\ldots, \calA^n$ that are all Pareto optimal and envy-free for the $n$ agents, respectively.\footnote{In particular, under allocation $\calA^i$ agent $i$ is envy-free against all other agents.} Here, proximity is quantified in terms of the symmetric differences between the agents' bundles: For any pair of these allocations $\calA^i=(A^i_1, \ldots, A^i_n)$ and $\calA^j=(A^j_1, \ldots, A^j_n)$, the union of the symmetric differences $\bigcup_{\ell \in [n]} \left( A^i_\ell \triangle A^j_\ell \right)$ (i.e., the set of reassigned items) is contained in a fixed set $R$. 

Notably, the relaxations of envy-freeness considered in discrete fair division capture the idea that a bounded number of changes in the bundles suffice to achieve envy-freeness. The notion considered in the current work, $\mathrm{EFR}$, conforms to this paradigm. Indeed, $\EF{1}$ requires that existing envy between any pair of agents $i$ and $j$ can be resolved by a hypothetical removal of an item form the bundle of $i$ or that of $j$. In $\mathrm{EFR}$ we instead consider reassignments of items from within a bounded-size set $R$ and, hence, obtain that at most $2|R|$ symmetric differences, across all the $n$ bundles, can ensure envy-freeness for any agent; see Definition \ref{defn:EFRk} and the subsequent remarks. Hence, instead of a per-agent-pair and per-bundle-removal guarantee, as in $\EF{1}$, under $\EFR{(n-1)}$, we obtain that two symmetric differences, on average across the $n$ bundles, suffice for achieving envy-freeness. 

Furthermore, we note that, from an $\EF{1}$ allocation of mixed manna, we can obtain an $\EFR{(n-1)}$ allocation $\calA$ (see the proof of Theorem \ref{thm:ERFexistence}). However, this implication, even when considered with $\PO$, fails to provide the above-mentioned feature of our result that the individual envy-free allocations, $\calA^1, \ldots, \calA^n$, are themselves $\PO$. In addition, we show that, when all the items are goods, then an $\EFR{\lfloor n/2 \rfloor}$ allocation always exists and can be computed efficiently\footnote{In this guarantee, we focus only on fairness and not on efficiency.} -- this result is interesting in and of itself and not implied by the  fact that $\EF{1}$ allocations exist for goods.  

Our contributions are listed next. 

\begin{itemize}
    \item We prove that $\EFR{(n-1)}$ and $\PO$ allocations always exist for mixed manna with additive valuations (\Cref{thm:EFRPO}). Our result utilizes the KKM theorem (\Cref{thm:kkm}), weighted welfare maximization for Pareto optimality, along with non-degeneracy and counting arguments.

    \item We also establish that when the number of agents, $n$, is fixed, an $\EFR{(n-1)}$ and $\PO$ allocation can be computed efficiently (\Cref{theorem:fixed-agents}).
\end{itemize}

\noindent 
We also obtain results for $\mathrm{EFR}$ by itself (and without the $\PO$ desideratum). 

\begin{itemize}
    \item We show that $\EFR{(n-1)}$ allocations for mixed manna with additive valuations  can be computed in polynomial time and this guarantee is compatible with $\EF{1}$ (\Cref{thm:ERFexistence}). 
For the case of chores, the bound of $(n-1)$ in $\mathrm{EFR}$ is tight. In particular, we exhibit instances (with only chores) wherein no allocation is $\EFR{(n-2)}$ (\Cref{thm:EFRnon-existence}). We also show that deciding whether a given allocation is $\EFR{k}$ is {\rm NP}-hard (\Cref{thm:ERFhardness}). Note that an analogous condition holds for Pareto optimality: One can find a $\PO$ allocation (under additive valuations) in polynomial time by, say, maximizing social welfare. However, determining whether a given allocation is $\PO$ is {\rm co-NP} hard \cite{de2009complexity}.  

    \item For goods, we prove that an $\EFR{\lfloor\frac{n}{2}\rfloor}$ allocations alway exist and can be computed efficiently (\Cref{thm:EFRgoods}). This result highlights an interesting dichotomy between goods and chores, since for the latter, even an $\EFR{(n-2)}$ allocation may fail to exist. Our algorithm for goods utilizes the round-robin framework with a judiciously chosen picking sequence over the agents.
Also, the $\lfloor\frac{n}{2}\rfloor$ bound for goods is tight (\Cref{thm:EFRgoodstight}). 
\end{itemize}

\subsection{Additional Related Work} 
We next discuss relevant prior works. \\

\noindent {\bf \(\EF{1}\) for Goods, Chores, and Mixed Manna.} 
In the setting of goods and monotone valuations, an $\EF{1}$ allocation always exists and can be computed efficiently by the envy-cycle elimination algorithm of Lipton et al.~\cite{LMM04}. Further, under additive valuations, the round-robin algorithm finds $\EF{1}$ allocations for both goods-only and chores-only settings . 

For mixed manna with additive valuations, the existence and efficient computation of \(\EF{1}\) allocations was established in \cite{azizetalFairallocation22} via a a double round-robin method. Bhaskar et al.~\cite{BSV21} noted that the standard envy-cycle elimination fails to give an $\EF{1}$ allocation for chores. Bypassing this barrier, they showed that there always exists a specific envy-cycle (specifically, a {top-trading envy cycle}), that can be resolved to efficiently compute an $\EF{1}$ allocation for chores. \\  

\noindent
{\bf $\EF{1}$ and $\PO$ for Mixed Manna.} For mixed manna, the existence of $\EF{1}$ and $\PO$ allocations remains open even among three agents with additive valuations. For two agents, Aziz et al.~\citep{azizetalFairallocation22} showed the existence and efficient computation of such allocations using a discrete version of the adjusted winner rule \cite{BT96fair}. Shoshan et al.~\cite{SHS23} studied a relaxation $\EF{[1,1]}$ and showed that, for two agents and under category constraints, $\EF{[1,1]}$ and $\PO$ allocations exist and can be computed efficiently. \\

\noindent
{\bf $\EF{1}$ and $\PO$ for Chores.} For chores with additive valuations, $\EF{1}$ and $\PO$ allocations are known to exist for two agents~\cite{azizetalFairallocation22}, three agents~\cite{gargetalNewAlgorithms23}, and three types of agents \cite{gargetalWeightedEF124}. As mentioned previously, the existence of $\EF{1}$ and $\PO$ allocations of chores is open even for four agents. Considering multiplicative approximations for the chores setting, \cite{GargConstantFactorSTOC25} shows the existence of $2$-$\EF{2}$ and $\PO$ along with $(n{-}1)$-$\EF{1}$ and $\PO$ allocations. 

Prior works have also obtained fairness and efficiency guarantees for restricted classes of valuations~\cite{ebadianetalHowfairly22,gargetalFairEfficientAAAI22,wuetalWeightedEF123,GargConstantFactorSTOC25,azizetalFairallocation23}. \\

\noindent
{\bf Fractional Allocations.}
In the context of divisible mixed manna (i.e., when the items can be fractionally assigned among the agents), Bogomolnaia et al.~\cite{Bogo2017} show that, under concave, homogeneous, and continuous valuations, a competetive equilibrium with equal incomes (CEEI) exists and that it is envy-free and $\PO$.

Building on this result, Sandomirskiy and Segal-Halevi \cite{SS22Sharing} proved, for additive valuations, the existence of envy-free and $\PO$ allocations in which at most $(n-1)$ items are fractionally assigned (and the remaining items are integrally allocated). We note that this result is incomparable to the guarantee obtained in the current work; Appendix \ref{sec:EFRvsEFSharing} shows that one cannot always `round' $(n-1)$ fractionally assigned chores to obtain an $\EFR{(n-1)}$ (integral) allocation. \\

\noindent
{\bf Concurrent Work.} In work independent of ours, Igarashi and Meunier also study fair and efficient allocation of mixed manna; see arXiv preprint \cite{IM25}. We also note that, for achieving Pareto efficiency in the indivisible items setting, weighted welfare maximization is utilized in \cite{IM25} and \cite{SHS23}.   
As in the current paper, \cite{IM25} provides a novel application of the KKM theorem in the discrete fair division context. \cite{IM25} considers fair division with a focus on category constraints and obtains a bound of $n^2 - n$ on the number of items that must be reallocated to achieve envy-freeness. The current paper establishes a stronger bound of $(n-1)$, albeit in the unconstrained setting. While the ideas used here to invoke the KKM Theorem are similar to those in \cite{IM25}, the techniques and the resulting $\mathrm{EFR}$ bounds differ quantitatively. In particular, there are notable differences in the approaches: \cite{IM25} perturbs the objective function and relies on rank lemmas for  linear programs, whereas we perturb the valuations and utilize the resulting non-degeneracy. Our approach allows us to employ counting and graph-theoretic arguments. Further, in contrast to this work, \cite{IM25} does not address polynomial-time computation of $\EFR{(n-1)}$ allocations of mixed manna and $\EFR{\lfloor n/2 \rfloor}$ allocations of goods. 
\section{Notation and Preliminaries}
\label{sec:preliminaries}
We study the problem of allocation $m \in \mathbb{Z}_+$ indivisible items (mixed manna) among $n \in \mathbb{Z}_+$ agents with additive valuations. The set of items and agents will be denoted by $[m]=\{1,\ldots, m\}$ and $[n]=\{1, \ldots, n\}$, respectively. Further, we will write $v_i(t) \in  \mathbb{R}$ to denote the valuation of each agent $i \in [n]$ for every item $t \in [m]$. Here, any item $t$ with a negative valuation, $v_i(t) < 0$, will be referred to as a chore for agent $i$, otherwise if $v_i(t) \geq 0$, then $t$ is a good for agent $i$. Under additive valuations, agent $i$'s value for each subset $S \subseteq [m]$ is equal to the sum of the values of the items in $S$, i.e., $v_i(S) = \sum_{t \in S} v_i(t)$.  A fair division instance will be specified by a tuple $\langle [n], [m], \{ v_i \}_{i \in [n]} \rangle$.

An allocation $\calA = (A_1,\ldots, A_n)$ is an $n$-partition of the set of items $[m]$; here, $A_i \subseteq [m]$ denotes the subset of items assigned to agent $i \in [n]$ and is referred to as $i$'s bundle. Note that for any allocation $\calA = (A_1,\ldots, A_n)$ we have $\cup_{i=1}^n A_i = [m]$ and $A_i \cap A_j  = \emptyset$, for $i \neq j$. We will write $\Pi_n(m)$ to denote the collection of all $n$-partitions of $[m]$, i.e., the collection of all allocations. \\

\noindent
{\bf Fairness Notions.} We next define the fairness criteria addressed in this work. We start with envy-freeness and then present its relaxations for the indivisible items setting. An allocation $\calA = (A_1, \ldots, A_n)$ is said to be envy-free if every agent values its own bundle at least as much as that of anyone else: $v_i(A_i) \geq v_i(A_j)$, for all $i, j \in [n]$. Further, we say that an allocation $\calA=(A_1, \ldots, A_n)$ is envy-free (specifically) for agent $i \in [n]$ if $i$ does not envy any other agent under $\calA$, i.e., $v_i(A_i) \geq v_i(A_j)$ for all $j \in [n]$.  

In the context of indivisible items, an envy-free allocation is not guaranteed to exist, and hence, multiple relaxations of envy-freeness have been considered in discrete fair division. A prominent notion here is envy-freeness up to one item ($\EF{1}$). This notion for mixed manna (i.e., when the indivisible items to be allocated include both goods and chores) is defined as follows; see, e.g., \cite{azizetalFairallocation22}.   

\begin{definition}[Envy-freeness up to one item]
 An allocation $\calB=(B_1, \ldots, B_n)$ is said to be \emph{envy-free up to one item} ($\EF{1}$) if for every pair of agents $i, j \in [n]$, either $i$ does not envy $j$, or there exists an item $t \in B_i \cup B_j$ with the property that $v_i(B_i \setminus \{t\}) \geq v_i (B_j \setminus \{t\})$. 
\end{definition}

Note that in any $\EF{1}$ allocation, $\calB$, existing envy between any pair of agents, $i, j \in [n]$, can be eliminated by the (hypothetical) removal of a chore from $B_i$ or a good from $B_j$. We address a complementary relaxation of envy-freeness wherein (instead of pair-specific modifications) a single, common set $R \subseteq [m]$---of at most $k$ items---serves as the basis for eliminating envy for all the agents simultaneously; here, $k$ is a defining parameter. Each agent may consider (hypothetically) reallocating items from $R$ differently to achieve envy-freeness for itself, but all such reassignments remain confined to $R$. The smaller the size of this set, $k$, the stronger the fairness guarantee. In particular, if $k =0$, then we recover exact envy-freeness. 

\begin{definition}[Envy-freeness up to $k$ reallocations] 
\label{defn:EFRk}
An allocation $\calA=(A_1, \ldots, A_n)$ is said to be envy-free up to $k \in \mathbb{Z}_+$ reallocations ($\EFR{k}$) if there exists a size-$k$ subset $R \subseteq [m]$ such that, for each agent $i\in [n]$, we can reassign items from within $R$ in $\A$ and obtain an allocation $\calA^i=(A^i_1, \ldots, A^i_n)$ which is envy-free for $i$.
\end{definition}
Note that, for an $\EFR{k}$ allocation $\calA$ and each allocation $\calA^i$, that is envy-free for agent $i$, it holds that $\bigcup_{\ell=1}^n \left( A_\ell \triangle A^i_\ell \right) \subseteq R$. That is, the symmetric differences between the bundles in $\calA$ and $\calA^i$ satisfy $\sum_{j=1}^n \left| A_j \triangle A^i_j  \right| \leq 2|R| = 2k$. This work establishes the existence of $\EFR{k}$ allocations with $k = (n-1)$. Hence, under such allocations $\calA$, for each agent $i$ one can obtain envy-freeness by reallocating at most two items on average across the $n$ bundles: $\frac{1}{n} \sum_{j=1}^n \left| A_j \triangle A^i_j  \right| \leq 2 - \frac{2}{n}$. \\

\noindent
{\bf Envy Graph and Envy Cycle.} For an allocation $\calA=(A_1, \ldots, A_n)$, the {envy graph} is a directed graph $G_\calA$ with $n$ vertices, one for each agent $i \in [n]$, and the graph contains the directed edge $(i,j) \in [n] \times [n]$ iff agent $i$ envies $j$ under the allocation, i.e., iff $v_i(A_i) <  v_i(A_j)$.  A directed cycle in $G_\calA$ is referred to as an {envy cycle}. \\

 \noindent
{\bf Pareto Efficiency and Social Welfare.} We next define the efficiency and welfare notions considered in this work. An allocation $\mathcal{X} = (X_1, \ldots, X_n)$ is said to be Pareto dominated by another allocation $\mathcal{Y} = (Y_1, \ldots, Y_n)$ if $v_i(X_i) \leq v_i(Y_i)$, for all agents $i \in [n]$, and one of these inequalities is strict.  An allocation $\calA$ is said to be {Pareto optimal} ($\PO$) if there does not exist any allocation $\mathcal{Y} \in \Pi_n(m)$ that Pareto dominates $\calA$. The following proposition notes that, for any Pareto optimal allocation $\calA$, the envy graph $G_\calA$ is acyclic. 

\begin{proposition}
\label{obs:cyclenotPO}
Let $\calA$ be a Pareto optimal allocation. Then, its envy graph $G_\A$ does not contain an envy cycle and, hence, there exists an agent $i\in [n]$ with the property that $\calA$ is envy-free for $i$.
\end{proposition}
\begin{proof}
Assume, towards a contradiction, that the envy graph $G_\A$ contains a directed cycle $C = i_1 \rightarrow i_2 \rightarrow \cdots \rightarrow i_k \rightarrow i_1$. Consider the allocation $\B=(B_1,\ldots, B_n)$ obtained by reassigning bundles in reverse order of $C$. That is, for each agent $i_\ell$ in the cycle, set $B_{i_{\ell}} =  A_{i_{\ell+1}}$; here, $i_{k+1} = i_1$. Also, for all agents $j$ not in the cycle, set $B_j = A_j$.  Since each agent in the cycle envies the next one, we have $v_{i_\ell}(B_{i_\ell}) = v_{i_\ell}(A_{i_{\ell+1}}) > v_{i_\ell}(A_{i_\ell})$. That is, the valuation of each agent in the cycle strictly increases, while for all other agents the valuations remain unchanged. Therefore, $\B$ Pareto dominates $\A$. Since this contradicts the Pareto optimality of $\A$, we obtain that $G_\calA$ does not contain an envy cycle. 

Further, the fact that $G_\calA$ is acyclic implies that there exists an agent $i \in [n]$ with no outgoing (envy) edge. That is, $G_\calA$ admits a sink node $i$. By construction of the envy graph, such an agent $i$ is envy free under $\calA$. 
\end{proof}

We will utilize weighted social welfare maximization as a means to achieve Pareto optimality. In particular,  for any allocation $\calA=(A_1, \ldots, A_n)$ and nonnegative weight vector $w =(w_1, \ldots w_n)^\intercal \in \Rn$, the weighted social welfare $\SW(\A,w)$ is defined as $\SW(\A,w) \coloneqq \sum_{i\in[n]} w_i \ v_i(A_i)$. 

Further, given any nonnegative weight vector $w =(w_1, \ldots w_n)^\intercal  \in \Rn$, the {support} of $w$ is defined as the components $i \in [n]$ for which the corresponding weights $w_i$ are strictly positive, $\supp(w) \coloneqq  \{i \in [n] : w_i > 0\}$. Note that, for any weight vector $w \in \Rn$ with positive components ($w_i > 0$ for all $i \in [n]$), an allocation that maximizes weighted social welfare, $\calA \  \in \argmax\limits_{(X_1, \ldots, X_n) \in \Pi_n(m)} \ \sum_{i\in [n]} w_i v_i(X_i)$, is Pareto optimal.  \\

\noindent 
{\bf Knaster-Kuratowski-Mazurkiewicz Theorem.} Next we provide the constructs required to state the Knaster-Kuratowski-Mazurkiewicz (KKM) Theorem. Write $\Delta_{n-1}$ to denote the standard $(n-1)$-dimensional simplex in $\mathbb{R}^n$, i.e., $\Delta_{n-1} = \left\{w \in \mathbb{R}^n : \sum_{i=1}^n w_i =1 \text{ and } w_i \ge 0 \text{ for all } i\in [n]\right\}$. 

For each nonempty $J\subseteq [n]$, write $\Delta_J$ to denote the convex hull of the basis vectors $\left\{e_j \in \mathbb{R}^n \right\}_{j \in J}$; we have, $\Delta_J = \left\{ w \in \Delta_{n-1} : w_i=0 \text{ for all } i \in [n]\setminus J \right\}$. Note that, for all vectors $w \in \Delta_J$, it holds that $\supp(w) \subseteq J$.

The KKM Theorem is a classic result in combinatorial topology. The theorem asserts that any $n$ closed subsets of $\Delta_{n-1}$ that satisfy a covering condition necessarily intersect. Formally, 
\begin{theorem}[\cite{knasteretalBeweisFixpunktsatzes29}]\label{thm:kkm}
Let $\mathcal{K}_1,\mathcal{K}_2, \ldots, \mathcal{K}_n$ be closed subsets of $\Delta_{n-1}$ with the property that $\Delta_J\subseteq\bigcup_{j \in J} \mathcal{K}_j$, for every nonempty $J \subseteq [n]$. Then, there exists a point $w^*\in \bigcap_{i=1}^n \mathcal{K}_i$.
\end{theorem}
Expressing the requirement in this theorem, we will say that subsets $\mathcal{K}_1, \ldots, \mathcal{K}_n \subseteq \Delta^{n-1}$ form a KKM covering if $\Delta_J \subseteq \bigcup_{j \in J} \mathcal{K}_j$ for every nonempty $J\subseteq [n]$.
\section{Fair and Efficient Allocation of Mixed Manna}\label{sec:efr_and_po}
This section shows that discrete fair division instances $\langle [n], [m], \{ v_i \}_{i \in [n]} \rangle$ with mixed manna and additive valuations always admit an $\EFRn$ and $\PO$ allocation. We utilize the KKM Theorem for this existential guarantee. Essentially, for each agent $i \in [n]$, we consider the set $\C_i \subseteq \Delta_{n-1}$ of weight vectors $w \in \Delta_{n-1}$ for which there exists a welfare-maximizing allocation, $\mathcal{X}^i$, that is envy-free for $i$. Here, for a complete analysis, additional technical ideas (specifically, value perturbations and weight shifting) are required. However, at a high level, the fact that the $\C_i$-s intersect (as ensured by the KKM Theorem) implies the existence of a weight vector $w^* \in \Delta_{n-1}$ under which, simultaneously for each agent $i \in [n]$, there exists an envy-free and (weighted) welfare-maximizing allocation, $\calA^i$. We will prove that these allocations $\{\calA^i \}_{i \in [n]}$ lead to the $\EFRn$ and $\PO$ guarantee. 

	To obtain a non-degeneracy condition (Definition \ref{defn:non-deg}), we will first perturb the given valuations by sufficiently small additive factors. In the analysis below, the weighted welfare maximization will be with respect to the perturbed valuations. Further, for each weight vector $w \in \Delta_{n-1}$ we will shift each component by a parameter $\eta >0$; this, in particular, will ensure that the welfare maximization is performed under weights that are all positive. Overall, our framework leverages the KKM theorem and perturbed welfare maximization (to ensure Pareto optimality) along with counting arguments to obtain the following main result. 
	\begin{restatable}{theorem}{EFRPO}
		\label{thm:EFRPO}
		Every fair division instance $\langle [n], [m],  \{ v_i \}_{i \in [n]} \rangle $, with indivisible mixed manna and additive valuations, admits an $\EFRn$ and $\PO$ allocation. 
	\end{restatable}	

\noindent
To establish the theorem for instance $\langle [n], [m],  \{ v_i \}_{i \in [n]} \rangle $, we will first construct perturbed valuations, $\{ \overline{v}_i \}_{i \in [n]}$, by subtracting (sufficiently small) positive random scalars, $\varepsilon_{i,t} >0$, from the underlying values $v_i(t)$. In particular, with a fixed constant $\varepsilon >0$ (specified below) in hand, we  draw $\varepsilon_{i,t}$ uniformly at random from the interval $(0, \varepsilon)$, independently for each $i \in [n]$ and $t \in [m]$, and set 
\begin{align}
\overline{v}_i(t) = v_i(t) - \varepsilon_{i,t} \label{eq:defn-v-bar} 
\end{align}

To select an appropriate $\varepsilon >0$, we define the following parameters 
\begin{align}
\lambda & \coloneqq \min\limits_{\substack{i \in [n], \ S, T  \ \subseteq [m] \\ S \cap T = \emptyset}}  \big\{   |v_i(S) - v_i(T)|  \  : \  v_i(S) \neq v_i(T)  \big\}  \label{eq:lambda} \\
\Lambda & \coloneqq \max\limits_{\substack{i \in [n], \ S, T \ \subseteq [m] \\ S \cap T = \emptyset}} \  |v_i(S) - v_i(T)|   \label{eq:Lambda} \\
\omega &  \coloneqq \min \limits_{\substack{\A,\B \in \Pin{m} \\  \SW(\A, {\mathbf 1}) \neq \SW(\B, {\mathbf 1}) }} \  \left| \SW (\A, {\mathbf 1}) - \SW (\B, {\mathbf 1}) \right| \label{eq:omega}
\end{align}
By definition, $\lambda$, $\Lambda$, and $\omega$ are all positive. Here, $\lambda$ captures the minimum envy that any agent can have, $\Lambda$ is the maximum envy, and $\omega$ is minimum distinct social welfare. 

We set the shift parameter $\eta$ as follows $0 < \eta \leq \frac{\lambda}{2 \Lambda n}$; note that $\eta \leq \frac{1}{2n}$. Furthermore, we set $0< \varepsilon < \frac{\eta }{4nm} \min\{  \lambda, \omega, 1 \}$.  
	 
\medskip

The analysis in this section involves both the given additive valuations $\{ v_i \}_{i}$ and the perturbed ones $\{\overline{v}_i \}_i$. Here, the perturbed valuations are also additive, i.e., for any subset of items $S \subseteq [m]$, we have $\overline{v}_i(S) = \sum_{t \in S} \overline{v}_i(t)$. We will also consider $\eta$-shifted weights and the corresponding social welfare under $\overline{v}_i$s. Specifically, for any allocation $\calA$ and weight vector $w \in \Delta_{n-1}$, write 
\begin{align*}
\overline{\SW}_\eta(\A, w) \coloneqq \sum_{i\in[n]} \ (w_i+\eta) \ \overline{v}_i (A_i)
\end{align*} 
Further, for any weight vector $w \in \Rn$, let $\overline{\SW}^*_\eta(w)$ denote the maximum welfare across all allocations, $\overline{\SW}^*_\eta(w) \coloneqq \max_{\mathcal{X} \in \Pin{m}} \ \overline{\SW}_\eta(\mathcal{X},w)$. Also, let $\overline{\Opt}_\eta(w)$ denote the set of allocations that maximize the $\overline{\SW}_\eta(\cdot, w)$. That is, the set of optimal allocations \[ \overline{\Opt}_\eta(w) \coloneqq \left\{\mathcal{X} \in \Pin{m} : \overline{\SW}_\eta(\mathcal{X}, w) = \overline{\SW}^*_\eta(w) \right\}.\] 

Equivalently, $\overline{\Opt}_\eta(w) = \argmax_{\mathcal{X} \in \Pi_n(m)} \ \overline{\SW}_\eta( \mathcal{X}, w)$. \\

Next, we will show that the perturbed valuations satisfy a useful non-degeneracy condition (Definition \ref{defn:non-deg}) with probability one. Here, $K_{n,m}$ denotes the complete bipartite graph with $n$ vertices in the left part and $m$ in the right; the vertices on the left correspond to the agents, and the ones on the right correspond to the items. 

\begin{definition}[Non-degenerate valuations]
\label{defn:non-deg}
For $n$ agents and $m$ items, values $\{ \widehat{v}_{i}(t) \}_{i \in [n], t \in [m]}$ are said to be non-degenerate if they bear the following two properties 
\begin{enumerate}[label=(\roman*)]
\item $\widehat{v}_i(t)\neq 0$ for all agents $i\in [n]$ and all items $t\in [m]$, and 
\item For every cycle $C: (i_1, t_1, i_2, t_2, \ldots, i_k, t_k, i_1)$ in $K_{n,m}$---with $i_1, \ldots, i_k \in [n]$ and $t_1, \ldots, t_k \in [m]$---it holds that $\prod \limits_{\ell=1}^k \left( \frac{ \widehat{v}_{i_{\ell + 1}}(t_\ell)}{\widehat{v}_{i_{\ell}}(t_\ell)} \right) \neq 1$; here, we consider indices modulo $k$, i.e., $i_{k+1} = i_1$. 
\end{enumerate}
\end{definition}

\begin{lemma}
\label{lem:non-degenerate}
For any fair division instance $\langle [n], [m], \{ v_i \}_{i \in [n]} \rangle$ with mixed manna and additive valuations, the perturbed values $\{ \overline{v}_i(t) \}_{i,t}$ (as defined above) are non-degenerate with probability one.
\end{lemma}
\begin{proof} 
Condition (i) in Definition \ref{defn:non-deg} holds for the perturbed values iff $nm$ inequalities of the form $\overline{v}_i(t) \neq 0$ are satisfied. That is, for condition (i) to hold none of the equalities of the form $v_i(t) - \varepsilon_{i,t} = 0$ should hold for the drawn $\varepsilon_{i,t}$. Since $\varepsilon_{i,t}$s are independent, continuous random variables, none of these $nm$ (finitely many) equalities hold, with probability one. 

Further, note that condition (ii) is composed of $O(n^n m^n)$ inequalities.\footnote{The number of cycles in $K_{n,m}$ is $O(n^n m^n)$.} Therefore, for the condition to hold, the drawn $\varepsilon_{i,t}$s should \emph{not} satisfy equalities of the following form
\begin{align*}
\prod_{\ell=1}^k \left( \frac{v_{i_{\ell + 1}}(t_\ell) - \varepsilon_{i_{\ell + 1},t_\ell}}{v_{i_{\ell}}(t_\ell) - \varepsilon_{i_{\ell},t_\ell}} \right)= 1.
\end{align*}  
As before, using the fact that $\varepsilon_{i,t}$s are independent, continuous random variables, we obtain that, with probability one, none of these (finitely many) equalities hold. Overall, the conditions hold for the perturbed values with probability one, and the lemma stands proved. 
\end{proof}

The following lemma shows that any allocation that maximizes the weighted social welfare allocation under the perturbed valuations, $\{\overline{v}_i\}_i$, continues to be Pareto optimal with respect to the underlying valuations $\{v_i \}_i$. Therefore, to identify a $\PO$ allocation for the given instance, it suffices to maximize the weighted social welfare under the perturbed valuations. Here, we will utilize the fact that the selected parameters $\varepsilon$ and $\eta$ are sufficiently small. 

\begin{lemma}\label{lem:POwrtInput}
For any weight vector $w\in \Delta_{n-1}$, let  $\calA$ be an allocation that maximizes shifted social welfare under perturbed valuations, i.e., $\A \in \overline{\Opt}_\eta(w)$. Then, $\A$ is $\PO$ with respect to the underlying valuations $\{ v_i \}_{i \in [n]}$.
\end{lemma}
\begin{proof}
Assume, towards a contradiction, that there exists an allocation $\B$ which Pareto dominates $\A$ under $v_i$s. It follows that $\sum_{i=1}^n v_i(B_i)> \sum_{i=1}^n v_i(A_i)$. Hence, by the definition of $\omega$, we have   
\begin{align}\label{eq:PO1}
\sum_{i=1}^n v_i(B_i)\geq \sum_{i=1}^n v_i(A_i) + \omega.
\end{align}
Since $\A \in \overline{\Opt}_\eta(w)$ it holds that  $\sum_{i=1}^n  (w_i+\eta) \ \overline{v}_i (A_i) \geq \sum_{i=1}^n  (w_i+\eta) \ \overline{v}_i (B_i) $.  Recalling the definition of the perturbed valuations we get 
\begin{align}
&\sum_{i=1}^n  (w_i+\eta) \left( v_i (A_i) - \sum_{t\in A_i} \varepsilon_{i,t}\right) \geq \sum_{i=1}^n  (w_i+\eta) \left( v_i (B_i) - \sum_{t\in B_i} \varepsilon_{i,t}\right) \notag 
\end{align}
Rearranging we obtain 
\begin{align}
\sum_{i=1}^n  (w_i+\eta) \left(\sum_{t\in B_i} \varepsilon_{i,t} - \sum_{t\in A_i} \varepsilon_{i,t} \right) &\geq \sum_{i=1}^n  (w_i+\eta) \left( v_i (B_i) - v_i(A_i) \right) \notag  \\
& = \sum_{i=1}^n  w_i \left( v_i (B_i) - v_i(A_i) \right) +  \sum_{i=1}^n  \eta \left( v_i (B_i) - v_i(A_i) \right) \notag \\
& \geq  \eta \sum_{i=1}^n  \left( v_i (B_i) - v_i(A_i) \right) \notag \\
& \geq \eta \omega \label{eq:PO2}
\end{align}
Here, the second inequality follows from the assumption that $\B$ Pareto dominates $\A$ (in particular, $v_i(B_i)\geq  v_i(A_i)$ for all $i$), and the third inequality from equation~(\ref{eq:PO1}). 

Next, recall that, for all $i \in [n]$ and $t \in [m]$, the random perturbations satisfy $0 < \varepsilon_{i,t} < \varepsilon$. Using these inequalities, we upper bound the left-hand-side of equation (\ref{eq:PO2}) as follows 
\begin{align}
\sum_{i=1}^n  (w_i+\eta) \left(\sum_{t\in B_i} \varepsilon_{i,t} - \sum_{t\in A_i} \varepsilon_{i,t} \right) & \leq \sum_{i=1}^n  (w_i+\eta) \left( \sum_{t\in B_i} \varepsilon_{i,t} \right) \nonumber \\
& \leq  \sum_{i=1}^n  (w_i+\eta) \ \varepsilon m \nonumber \\
& = \varepsilon m \left( 1 + \eta n \right) \tag{since $\sum_{i=1}^n w_i = 1$} \\
& \leq \frac{\eta}{4nm} \omega \ m \left( 1 + \eta n \right) \tag{since $\varepsilon<  \frac{\eta }{4nm} \omega$} \\
& = \frac{\eta}{4n} \omega  \ \left( 1 + \eta n \right) \label{eq:PO3}
\end{align}
Equations (\ref{eq:PO2}) and (\ref{eq:PO3}) imply $\frac{\eta \omega}{4n} (1+n \eta) \geq \eta \omega$. Hence, we obtain $\eta \geq \frac{4n-1}{n} >1$. This bound, however, contradicts the fact that $\eta \leq \frac{1}{2n} < 1$. Therefore, by way of contradiction, we obtain that allocation $\calA$ is not Pareto dominated by any other allocation under the valuations $\{v_i\}_i$. This completes the proof of the lemma. 
\end{proof}

The following lemma provides another useful property of the perturbed valuations $\overline{v}_i$s: the perturbations preserve envy. 
\begin{lemma}
\label{lemma:envy-preserved}
For any agent $i \in [n]$ and any pair of disjoint subsets $S, T \subseteq [m]$, if $v_i(S) > v_i(T)$, then $\overline{v}_i(S) \geq \overline{v}_i(T) + \frac{\lambda}{2}$. 
\end{lemma}
\begin{proof}
Since $v_i(S) > v_i(T)$, the definition of $\lambda$ gives us $v_i(S) \geq v_i(T) + \lambda$. In addition, the perturbed values of these sets satisfy
\begin{align}
\overline{v}_i(S) - \overline{v}_i(T) \geq v_i(S) - \varepsilon |S| - v_i(T) \geq v_i(S) - \varepsilon m - v_i(T) \geq \lambda - \varepsilon m  \label{ineq:lamb-eps}
\end{align}
By definition, $\varepsilon < \frac{\eta}{4nm} \lambda \leq \frac{1}{8n^2 m} \lambda$; recall that $\eta \leq \frac{1}{2n}$. Hence, inequality (\ref{ineq:lamb-eps}) reduces to 
\begin{align*}
\overline{v}_i(S) - \overline{v}_i(T) \geq \lambda - \varepsilon m \geq \lambda - \frac{\lambda}{8n^2 m} m = \lambda \left( 1 - \frac{1}{8n^2} \right) \geq \frac{\lambda}{2}.
\end{align*}
The lemma stands proved. 
\end{proof}

\noindent 
{\bf KKM Covering.} For each agent $i \in [n]$, we define set 
\begin{align}
\C_i \coloneqq \{w \in \Delta_{n-1} :  \text{there exists an allocation } \mathcal{X}^i \in \overline{\Opt}_\eta(w) \text{ such that } i \text{ is envy-free under } \mathcal{X}^i\} \label{eqn:def-Ci}
\end{align}
Note that in the definition of $\C_i$ the envy-freeness of agent $i$ under allocation $\mathcal{X}^i$ is considered with respect to the underlying valuation $v_i$.  

The following two lemmas show that $\C_i$s are closed and form a KKM covering, respectively. Hence, we can invoke the KKM Theorem with sets to obtain that these sets intersect. We will utilize this guaranteed intersection to establish Theorem \ref{thm:EFRPO} in Section \ref{section:proof-efrpo}.

\begin{lemma}\label{lem:closed}
For every \(i\in[n]\), the set \(\C_i\) is closed.
\end{lemma}
\begin{proof}
 Let \(\{w^k\}_{k\in\mathbb{N}}\) be any convergent sequence in \(\C_i\) and $\overline{w} \in \Delta_{n-1}$ be the limit point of the sequence. To show that $\C_i$ is closed it suffice to show that the limit point $\overline{w} \in \C_i$. 
    
 Note that, for each $k \in \mathbb{N}$, the vector $w^k\in \C_i $ and, hence, there exists an allocation \(\mathcal{X}^i_k \in \overline{\Opt}_\eta(w^k) \) with the property that agent $i$ is envy-free under $\mathcal{X}^i_k$. Since the collection of allocations is finite, by the infinite pigeonhole principle, some allocation, \(\widehat{\mathcal{X}}^i\), appears infinitely often in the sequence. Let \(\{w^{k_j}\}_{j \in \mathbb{N}} \) be the corresponding infinite subsequence for which $\widehat{\mathcal{X}}^i \in \overline{\Opt}_\eta(w^{k_j})$ for all $j \in \mathbb{N}$. Since \(\{w^{k_j}\}_j \) is an infinite subsequence of the convergent sequence \(\{w^k\}_k \), the subsequence  also converges to \(\overline{w}\).

Next, note that, as a function of $w \in \Delta_{n-1}$, the optimal social welfare \( \overline{\SW}^*_\eta(w)\) is continuous. Indeed, the function \(\overline{\SW}^*_\eta(w)\) is defined as the maximum of a finite number of continuous functions, $\overline{\SW}_\eta(\cdot, w)$, each of which is a linear function of \(w\). Since the maximum of continuous functions is also continuous, \(\overline{\SW}^*_\eta(w)\) is continuous. Therefore, via the sequential criterion for continuity, we have \(\overline{\SW}^*_\eta(w^{k_j}) \to \overline{\SW}^*_\eta(\overline{w})\) as \(j \to \infty\). Similarly, \(\overline{\SW}_\eta(\widehat{\mathcal{X}}^i, w^{k_j}) \to \overline{\SW}(\widehat{\mathcal{X}}^i,  \overline{w})\) as \(j \to \infty\).

The sequence \( \left\{\overline{\SW}_\eta(\widehat{\mathcal{X}}^i , w^{k_j}) \right\}_j \) is the same as \( \left\{\overline{\SW}^*_\eta(w^{k_j}) \right\}_j \), since \(\widehat{\mathcal{X}}^i \in \overline{\Opt}_\eta(w^{k_j})\) for all \(j\). Hence,  \(\overline{\SW}_\eta(\widehat{\mathcal{X}}^i, \overline{w}) = \overline{\SW}^*_\eta(\overline{w})\). This equality implies that \(\widehat{\mathcal{X}}^i \in \overline{\Opt}_\eta(\overline{w})\). Given that agent \(i\) is envy-free under 
\( \widehat{\mathcal{X}}^i \), we obtain that \(\overline{w} \in \C_i\). Therefore, the set \(\C_i\) is closed, and this completes the proof of the lemma. 
\end{proof}

The lemma below shows that $\C_i$s satisfy the KKM covering condition. 
\begin{lemma}\label{lem:covering}
For each nonempty  \(J\subseteq[n]\) it holds that $\Delta_J \subseteq \bigcup_{j \in J} \mathcal{C}_j$. 
\end{lemma}
\begin{proof}
We establish the lemma by considering two complementary and exhaustive cases --  Case {\rm I}: The set $J = [n]$ and Case {\rm II}: The set $J \subsetneq [n]$.  \\

\noindent
{\it Case {\rm I}: Set $J = [n]$.}  Consider any weight vector \(w \in \Delta_{[n]} = \Delta_{n-1}\). We will show that $w \in \bigcup_{j \in [n]} \ \C_j$ and, hence, obtain that the covering condition holds for $J = [n]$. For the vector $w$, let \(\calA\) be an allocation in \(\overline{\Opt}_\eta(w)\). Lemma~\ref{lem:POwrtInput} implies that $\calA$ is $\PO$ with respect to the underlying valuations $\{v_i\}_i$. Further, Proposition \ref{obs:cyclenotPO} ensures that there exists an agent \(\hat{i}\in[n]\) that is envy-free under \(\A\). Therefore, \(w\in \C_{\hat{i}}\), i.e., as desired, we have $w \in \bigcup_{j \in [n]} \ \C_j$. This completes the analysis of Case {\rm I}. \\

\noindent
{\it Case {\rm II}:  Nonempty set \(J \subsetneq [n]\).} Consider any vector \(w\in \Delta_J\) and allocation  \(\A\in \overline{\Opt}_\eta(w)\). In Claim \ref{claim:envy-free-agent} below we will show that there exists an agent $\widehat{j} \in \supp(w)\subseteq J$ that is envy-free under \(\A\).\footnote{Here, as in the definition of $\mathcal{C}_i$s, envy-freeness is considered with respect to the underlying valuations $\{v_i\}_i$.} For such an agent $\widehat{j}$, we have $w \in \C_{\widehat{j}}$ and, hence, $w \ \in \bigcup\limits_{j \in \supp(w)} C_j \subseteq \bigcup\limits_{j \in J} C_j$. Since this containment holds for all \(w\in \Delta_J\), the KKM covering condition is satisfied in this case as well. 

\begin{claim}
\label{claim:envy-free-agent}
There exists an agent $\widehat{j} \in \supp(w)\subseteq J$ that is envy-free under \(\A\).
\end{claim}
\begin{proof}
Assume, towards a contradiction, that every agent in \(\supp(w)\) envies some other agent under allocation \(\A\), i.e., there does not exist a $\widehat{j} \in \supp(w)$ that is envy-free under \(\A\). Equivalently,  every agent $j \in \supp(w)$ has at least one outgoing edge in the envy-graph $G_\calA$. Now, Lemma~\ref{lem:POwrtInput} ensures that $\A$ is $\PO$ with respect to the valuations $\{v_i\}_i$. Hence, $G_\calA$ is acyclic (Proposition \ref{obs:cyclenotPO}). Also, the fact that $w \in \Delta_{n-1}$ implies that there exists an agent $j_1 \in \supp(w)$ with weight $w_{j_1} \geq 1/n$.  

Using the above-mentioned observations about the envy-graph $G_\calA$, we obtain that there exists a directed path \(P = j_1 \rightarrow  j_2  \rightarrow \ldots \rightarrow j_p \rightarrow j_{p+1}\) in the graph $G_\calA$ with \(j_1, \ldots, j_p \in \supp(w)\) and \(j_{p+1} \in [n]\setminus \supp(w)\). For ease of notation, we reindex these agents as $1, 2, \ldots, p \in \supp(w)$ and $(p+1) \in [n] \setminus \supp(w)$. That is, we reindex such that the path $P = 1 \rightarrow 2 \rightarrow \ldots \rightarrow p \rightarrow (p+1)$ and we have $w_1 \geq 1/n$. Also, given that the path $P$ exists in the envy-graph $G_\calA$, for each $i \in \{1, 2, \ldots, p\}$, agent $i$ envies agent $(i+1)$, i.e., $v_i(A_i) < v_i(A_{i+1})$. 

We reassign bundles in $\calA$ along the path $P$ to obtain another allocation $\calB$. In particular, for all $1 \leq i \leq p$, set $B_i = A_{i+1}$, along with $B_{p+1} = A_1$. For all the remaining agents the assigned bundles remain unchanged. We will next show that $\overline{\SW}_\eta(\calB, w) > \overline{\SW}_\eta(\A, w)$; this contradicts the welfare-optimality of $\calA$ (i.e., contradicts \(\A\in \overline{\Opt}_\eta(w)\)). Hence, via contradiction, we will obtain that there exists an agent $\widehat{j} \in \supp(w) \subseteq J$ that is envy-free under \(\calA\). 

The difference in weighted social welfare of the allocations $\calB$ and $\calA$ satisfies 
\begin{align}
\overline{\SW}_\eta(\calB, w) - \overline{\SW}_\eta(\A, w) & = \sum_{i=1}^n (w_i + \eta) \overline{v}_i (B_i) \  - \ \sum_{i=1}^n (w_i + \eta) \overline{v}_i(A_i) \notag \\
& = \sum_{i=1}^{p+1} (w_i + \eta) \big( \overline{v}_i (B_i) - \overline{v}_i (A_i) \big) \notag \\
& = \sum_{i=1}^p (w_i + \eta) \big( \overline{v}_i (A_{i+1}) - \overline{v}_i (A_i) \big) \ +  \left( w_{p+1} + \eta \right) \left( \overline{v}_{p+1} (A_1) - \overline{v}_{p+1} (A_{p+1}) \right) \label{eq:calB}
\end{align}
The last equality follows from the construction of $\calB$. Also, note that $w_{p+1} = 0$, since $(p+1) \notin \supp(w)$. Hence, equation (\ref{eq:calB}) reduces to 
\begin{align}
\overline{\SW}_\eta(\calB, w) - \overline{\SW}_\eta(\A, w) & = \sum_{i=1}^p (w_i + \eta) \big( \overline{v}_i (A_{i+1}) - \overline{v}_i (A_i) \big) \ +  \ \eta \big( \overline{v}_{p+1} (A_1) - \overline{v}_{p+1} (A_{p+1}) \big) \label{eq:calBtwo}
\end{align}
As noted previously, for each $i \in \{1, 2, \ldots, p\}$, agent $i$ envies agent $(i+1)$. Hence, for each $i \in \{1, \ldots, p\}$, the inequality $v_i(A_{i+1}) > v_i(A_{i})$ and Lemma \ref{lemma:envy-preserved} give us  
$\overline{v}_i (A_{i+1}) > \overline{v}_i (A_i)$. Using these inequalities for $i \in \{2, \ldots, p\}$, we simplify equation (\ref{eq:calBtwo}) as follows 
\begin{align}
\overline{\SW}_\eta(\calB, w) - \overline{\SW}_\eta(\A, w) & \geq \ (w_1 + \eta) \big( \overline{v}_1 (A_{2}) - \overline{v}_1 (A_1) \big) \ + \  \eta \big( \overline{v}_{p+1} (A_1) - \overline{v}_{p+1} (A_{p+1}) \big) \notag \\
& \geq  (w_1 + \eta)  \frac{\lambda}{2} \ + \  \eta \big( \overline{v}_{p+1} (A_1) - \overline{v}_{p+1} (A_{p+1}) \big)  \tag{Lemma \ref{lemma:envy-preserved}} \\
& \geq \left( \frac{1}{n} + \eta \right) \frac{\lambda}{2}  \ + \  \eta \big( \overline{v}_{p+1} (A_1) - \overline{v}_{p+1} (A_{p+1}) \big) \tag{since $w_1 \geq 1/n$} \\
& = \frac{\lambda}{2n} \ + \ \eta \left(  \overline{v}_{p+1} (A_1) - \overline{v}_{p+1} (A_{p+1}) + \frac{\lambda}{2} \right) 
 \label{eq:calB3}
\end{align} 
Further, note that 
\begin{align}
\overline{v}_{p+1} (A_1) - \overline{v}_{p+1} (A_{p+1}) & \geq v_{p+1} (A_1) - \varepsilon |A_1| \ - \ v_{p+1}(A_{p+1}) \notag \\
& \geq v_{p+1} (A_1) - \ v_{p+1}(A_{p+1})  \ - \varepsilon m \notag \\ 
& \geq - \Lambda - \varepsilon m \tag{by definition of $\Lambda$} \notag \\
& > - \Lambda - \frac{\lambda}{2} \label{ineq:strict}
\end{align}
The last (strict) inequality follows from the bound $\varepsilon m < \left( \frac{\eta}{4nm} \lambda \right) m \leq \left( \frac{1}{8n^2 m} \lambda \right) m < \frac{\lambda}{2}$; recall that $\eta \leq \frac{1}{2n}$. 

Equations (\ref{eq:calB3}) and (\ref{ineq:strict}) lead to 
\begin{align}
\overline{\SW}_\eta(\calB, w) - \overline{\SW}_\eta(\A, w) & > \frac{\lambda}{2n} - \eta \ \Lambda \notag \\
& \geq \frac{\lambda}{2n} - \left( \frac{\lambda}{2 \Lambda n} \right)  \Lambda \tag{since $\eta \leq \frac{\lambda}{2 \Lambda n}$} \\
& = 0 \label{ineq:swmore} 
\end{align}
Equation (\ref{ineq:swmore}) contradicts the welfare-optimality of $\calA$. Hence, via contradiction, we will obtain that there exists an agent $\widehat{j} \in \supp(w) \subseteq J$ that is envy-free under \(\calA\). This completes the proof of the claim.
\end{proof}
As mentioned previously, Claim \ref{claim:envy-free-agent} implies that $w \in \C_{\widehat{j}}$ for agent $\widehat{j} \in \supp(w) \subseteq J$. Hence, we have the containment $w \in \cup_{j \in J} \ \C_j$ and  obtain KKM covering in the second case as well. 

The lemma stands proved. 
\end{proof}

Using the above-mentioned constructs and lemmas, we now establish \Cref{thm:EFRPO}.

\subsection{Proof of Theorem \ref{thm:EFRPO}}
\label{section:proof-efrpo}
As mentioned previously, for the given fair division instance $\langle [n], [m],  \{ v_i \}_{i \in [n]} \rangle$, we first construct the perturbed valuations $\{ \overline{v}_i \}_i$. \Cref{lem:non-degenerate} ensures that, with probability one, the perturbed valuations are non-degenerate. That is, there exist perturbations, $\varepsilon_{i,t}$s, for which $\overline{v}_i$s are non-degenerate. We will establish the theorem considering such non-degenerate $\overline{v}_i$s.  

Lemma~\ref{lem:closed} gives us that, for each agent $i \in [n]$, the set $\C_i$ (see equation (\ref{eqn:def-Ci})) is closed. Further, these sets uphold the KKM covering condition (Lemma~\ref{lem:covering}). Hence, applying the KKM theorem (Theorem~\ref{thm:kkm}), we obtain that there exists a vector \(w^*\in \bigcap_{i=1}^n \C_i\). 

The definition of the sets $\C_i$s imply that, under \(w^*\in \bigcap_{i=1}^n \C_i\), for each agent $i\in [n]$ there exists an allocation $\A^i  \in \overline{\Opt}_\eta(w^*)$ that is envy-free for $i$. Note that all these allocations $\A^i \in \overline{\Opt}_\eta(w^*)$ are Pareto optimal with respect to the underlying valuations (\Cref{lem:POwrtInput}). 
	
We set any one of these agent-specific envy-free allocations as the desired $\mathrm{EFR}$ allocation, say $\calA = \calA^1$; indeed, $\calA$ is $\PO$. 

We will next show that there exists a set of items $R(w^*) \subseteq [m]$, of size at most $(n-1)$, such that $\bigcup\limits_{\ell=1}^n  \left( A^i_\ell \triangle A^j_\ell \right) \subseteq R(w^*)$, for all $i, j \in [n]$. The existence of subset $R(w^*)$--with cardinality at most $(n-1)$---certifies that, as desired, the allocation $\calA$ is $\EFR{(n-1)}$.     

Now, for each item $t \in [m]$ write $D(t, w^*)$ to denote the set of agents $i$ for whom $(w^*_i+\eta)\ \overline{v}_i(t)$ is maximum among all the agents, $D(t, w^*) \coloneqq \left\{ i \in [n] \ : \ (w_i^* +\eta) \overline{v}_i(t) = \max\limits_{j\in [n]} \ (w_j^*+ \eta) \overline{v}_j(t) \right\}$. That is, $D(t, w^*) = \argmax_{j \in [n]} \ (w_j^* +\eta) \overline{v}_j(t)$. 

Recall that the valuations $\{\overline{v}_i \}_i$ are additive. Hence, in every welfare-maximizing allocation---i.e., in every allocation in $\overline{\Opt}_\eta(w^*)$---each item $t \in [m]$ must be assigned to an agent in $D(t, w^*)$. This property holds for every $\calA^i \in \overline{\Opt}_\eta(w^*)$. Further, if for an item $s \in [m]$ we have $|D(s, w^*)| = 1$, then in every welfare-maximizing allocation, $s$ is assigned to the unique agent $\ell \in D(s, w^*)$. In particular, for an item $s \in [m]$ with $|D(s, w^*)| = 1$, it holds that $s$ is assigned to the same agent across all the allocations $\{ \calA^i \}_{i \in [n]}$. We define $R(w^*) \subseteq [m]$ to be the set of complementary items, 
\begin{align*}
R(w^*) \coloneqq \left\{ t \in [m] \ : \ |D(t, w^*)|\geq 2  \right\}.
\end{align*}
Since each item $s \notin R(w^*)$ is assigned to the same agent in any pair of allocation $\calA^i, \calA^j \in \overline{\Opt}_\eta(w^*)$, we have the desired containment $\bigcup\limits_{\ell=1}^n  \left( A^i_\ell \triangle A^j_\ell \right) \subseteq R(w^*)$. 

We will now complete the proof by proving that $|R(w^*)| \leq n-1$. Towards this, consider a bipartite graph $\calH = ([n] \cup R(w^*), E)$; here, the left part consists of all the $n$ agents and the right part are the items in $R(w^*)$. We include edge $(i, t) \in [n] \times R(w^*)$ in the edge set, $E$, iff $i \in D(t, w^*)$. Since $|D(t,w^*)|\geq 2$ for each item $t \in R(w^*)$, the degree of each right-hand-side node in $\calH$ is at least $2$. The following claim shows that the graph $\calH$  is acyclic -- the claim utilizes the non-degeneracy of $\overline{v}_i$s.
	
\begin{claim} \label{claim:acyclic}
The bipartite graph $\calH = ([n] \cup R(w^*), E)$ is acyclic.
\end{claim}
\begin{proof}
Assume, towards a contradiction, that $\calH$ contains a cycle $C$. We reindex the agents and the items such that agents $1,2, \ldots, k$ and items $t_1, \ldots t_k \in R(w^*)$ constitute the cycle, i.e., $C = (1, t_1, 2, t_2, \ldots k, t_{k}, 1)$, for $2 \leq k \leq n$. Note that, for each $\ell \in [k-1]$, both the edges $(\ell, t_\ell)$ and $(\ell+1, t_\ell)$ are in $\calH$. Hence, agents $\ell, \ell+1 \in D(t_\ell, w^*)$, for each $\ell \in [k-1]$. Similarly, agents $1, k \in D(t_k, w^*)$. 

The definition of $D(t, w^*)$ gives us $\frac{\overline{v}_{\ell+1}(t_\ell)}{\overline{v}_\ell(t_\ell)} =\frac{w^*_\ell+\eta}{w^*_{\ell+1}+ \eta}$, for all $1 \leq \ell \leq k-1$, and $\frac{\overline{v}_{1}(t_k)}{\overline{v}_k(t_k)} =\frac{w^*_k+\eta}{w^*_{1}+ \eta}$. Note that these equalities are well-defined, since $\eta>0$ and $\overline{v}_i(t)$s are non-degenerate; in particular, $\overline{v}_i(t)\neq 0$ for all $i\in [n]$ and $t\in [m]$ (see Definition \ref{defn:non-deg}). Multiplying these equalities, and considering indices modulo $k$ (i.e., $\overline{v}_{k+1} (\cdot) = \overline{v}_1 (\cdot)$ and $w^*_{k+1} = w^*_1$), we obtain 
\begin{align} 
 \prod_{\ell\in [k]} \frac{\overline{v}_{\ell+1}(t_\ell)}{\overline{v}_\ell(t_\ell)} 	 & = \prod_{\ell\in [k]}  \frac{w^*_\ell+\eta}{w^*_{\ell+1}+ \eta}  = 1 \label{eq:prod-deg}
\end{align}

Equation (\ref{eq:prod-deg}), however, contradicts the non-degeneracy of $\overline{v}_i$s (see condition (ii) in Definition \ref{defn:non-deg}). Therefore, by way of contradiction, we obtain that $\calH$ is acyclic, and the claim stands proved. 
\end{proof}

Claim \ref{claim:acyclic} leads to the desired bound on the size of $R(w^*)$: Since $\calH = ([n] \cup R(w^*), E)$ is acyclic, the total number of edges in the graph $|E|  \leq n + |R(w^*)| -1$. Also, recall that each node in the right-hand-side of $\calH$ (i.e., each node in $R(w^*)$) has degree at least $2$. Hence, $|E| \geq 2 |R(w^*)|$. These inequalities give us $2|R(w^*)| \leq n + |R(w^*)| -1$. Therefore, the desired bound follows:  $|R(w^*)| \leq n-1$.

Overall, we have a $\PO$ allocation $\calA$ and a set of items $R(w^*)$---of size at most $(n-1)$---such that, for each agent $i\in [n]$, we can reassign items from within $R(w^*)$ in $\A$ and obtain an allocation $\calA^i$ which is envy-free for $i$. That is, $\calA$ is $\EFR{(n-1)}$ and $\PO$. The theorem stands proved. 
\section{Finding $\EFR{(n-1)}$ and $\PO$ Allocations for Fixed Number of Agents}
\label{section:fixed-agents}

This section establishes that, for any given fair division instance $\langle [n], [m],  \{ v_i \}_{i \in [n]} \rangle $ with mixed manna and additive valuations, an $\EFR{(n-1)}$ and $\PO$ allocation can be computed in 
$m^{\mathrm{poly}(n)}$ time. Hence, for a fixed number of agents, one can compute a fair and efficient allocation efficiently. In particular, we will build upon the existential guarantee obtained in \Cref{thm:EFRPO} to establish the following computational result. 

\begin{restatable}{theorem}{constantagents}
\label{theorem:fixed-agents}
For any given fair division instance $\langle [n], [m], \{v_i\}_{i=1}^n \rangle$ among a fixed number of agents that have additive valuations over mixed manna, we can compute an $\EFR{(n-1)}$ and $\PO$ allocation in polynomial time. 
\end{restatable}

First, note that in the input instance $\langle [n], [m], \{v_i\}_{i=1}^n \rangle$ the given valuations have bounded bit complexity. Hence, by scaling and without loss of generality, we can assume that the input values, $v_i(t)$s, are integers; this assumption will be utilized throughout this section. 

We will begin by efficiently finding perturbations $\varepsilon_{i,t}$s that conform to the bounds specified in Section \ref{sec:efr_and_po} and lead to perturbed values, $\overline{v}_i(t) = v_i(t) - \varepsilon_{i,t}$, that are non-degenerate (Definition \ref{defn:non-deg}).  \\

\noindent
{\bf Efficiently constructing non-degenerate valuations.} 
For the given integer values $v_i(t)$s, relevant bounds on parameters defined in equations (\ref{eq:lambda}), (\ref{eq:Lambda}), and (\ref{eq:omega}) can be efficiently obtained as follows. 
For each agent $i \in [n]$, write  $G_i \coloneqq \{t \in[m]  : v_i(t) \ge 0\}$ to denote the set of all goods (i.e., the set of all non-negatively valued items) for agent $i$. Conforming to equation (\ref{eq:Lambda}), here it holds that $\Lambda = \max_{i \in [n]} \left(v_i(G_i) - v_i([m] \setminus G_i) \right)$. Further, given that the input values are integers, we have $\lambda \geq 1$ and $\omega \geq 1$. With these parameters in hand, we select $\eta$ and $\varepsilon$ such that they have polynomial bit complexity and satisfy: $0 < \eta \leq \frac{1}{2 \Lambda n}$ and $0< \varepsilon < \frac{\eta }{4nm}$.   

We set the perturbations \(\varepsilon_{i,t}\), one at a time, in the order $i \in \{1,\ldots, n\}$ and $t \in \{1, \ldots, m\}$. For each $i$ and $t$, we consider all the cycles $C$ in the complete bipartite graph $K_{n,m}$ with the property that the edge \((i,t)\) is the only one in \(C\) whose perturbation \(\varepsilon_{i,t}\) has not  been set so far. The non-degeneracy condition (see Definition \ref{defn:non-deg}) for such a cycle is
\begin{align*}
\prod \limits_{\ell=1}^k \left( \frac{ {v}_{i_{\ell + 1}}(t_\ell) - \varepsilon_{i_{\ell+1},t_\ell} }{ {v}_{i_{\ell}}(t_\ell) -  \varepsilon_{i_\ell,t_\ell} } \right) \neq 1.
\end{align*}
Since all perturbations, except \(\varepsilon_{i,t}\), participating in this inequality have been fixed already, the inequality rules out a single, rational value for \(\varepsilon_{i,t}\). That is, for each considered cycle $C$, the non-degeneracy condition requires that \(\varepsilon_{i,t}\) is not equal to a specific value. Since the number of agents $n$ is fixed, there are a fixed number of (in particular, $O(m^n n^n)$) cycles in $K_{n,m}$. Hence, we can assign \(\varepsilon_{i,t}\) any value (with polynomial bit complexity) in the range $(0, \varepsilon)$, besides the ones that are ruled out. Repeating this for process for  each \((i,t)\) gives us perturbations under which the values $\overline{v}_i(t) = v_i(t) - \varepsilon_{i,t}$ are non-degenerate. Our algorithm works with such non-degenerate $\overline{v}_i$s. 

\medskip
  
From the proof of \Cref{thm:EFRPO}, we have that there exists an $\EFR{(n-1)}$ and $\PO$ allocation $\calA^*$ and a weight vector $w^* \in \Delta_{n-1}$ with the property that $\calA^*$ is weighted welfare maximizing under these weights, $\calA^* \in \overline{\Opt}_\eta(w^*)$. Further, there is a subset $R^*\subseteq [m]$---of size at most $(n-1)$---such that, for each agent, we can achieve envy-freeness by reassigning items from within $R^*$. 

Next, we will show that there exists a set of $O(n^2)$ items whose assignment in $\calA^*$ can be bootstrapped to identify the allocation of all the other items in this  $\EFR{(n-1)}$ and $\PO$ allocation. We will refer to these $O(n^2)$ items as {\it separating items}. At a high level, our algorithm exhaustively searches for these separating items by enumerating over all possible size-$O(n^2)$ subsets of items and their assignment among the agents. For each candidate subset and assignment, the algorithm then uses an (efficient) bootstrapping method to allocate all the remaining items and, hence, finds a tentative allocation $\calA$. Finally, the algorithm tests if $\calA$ is $\EFR{(n-1)}$ and $\PO$, or not. In the proof of the theorem, we will show that these tests can be performed in polynomial time (for a fixed number of agents) and the time complexity of the algorithm is $m^{\mathrm{poly}(n)}$.

Recall that, for each item $t$, the set of agents
$D(t, w^*) \coloneqq \argmax_{j \in [n]} \ (w_j^* + \eta)\overline{v}_j(t)$. Write $I^*_i \coloneqq \{ t \in A^*_i  : |D(t, w^*)|=1 \}$, i.e., $I_i^*$ is the set of items in $A^*_i$ with the property that $|D(t, w^*)|=1$. The construction of $R^*$ in the proof of \Cref{thm:EFRPO} gives us $I_i^*:= A^*_i \setminus R^*$.

For each pair of distinct agents $i, j \in [n]$, with $I^*_i \neq \emptyset$, we define \emph{separating items} $g^*_{ij}, c^*_{ij} \in I^*_i$ and (forbidden) set $F_{ij} \subseteq [m]$. Lemma \ref{cor:intersection} (stated below) asserts that, for each agent $i$, the set $I^*_i$ is equal to the intersection of $F_{ij}$s, with $j \neq i$. This result is the basis of the bootstrapping idea mentioned above: for any agent $i \in [n]$, given the separating items $\{g^*_{ij}, c^*_{ij} \}_{j \neq i}$, we can first efficiently identify the sets $F_{ij} \subseteq [m]$, for all $j \neq i$, and then take their intersection to find $I^*_i$.

For each $i \neq j$, write: 
\begin{itemize}
    \item $O^{+}_{ij}$ to denote items that are goods for both agents $i$ and $j$.
    \item $O^{-}_{ij}$ to denote items that are chores for both agents $i$ and $j$. 
    \item $Q_{ij}$ to denote items that are goods for $i$ and chores for $j$.
\end{itemize}
In addition, for $i \neq j$, if $I^* \cap O^+_{ij} \neq \emptyset$,  we define
 \begin{align}
g^*_{ij}  \in \ \argmax_{g \in I^*_i \cap O^{+}_{ij}} \frac{\overline{v}_j(g)}{\overline{v}_i(g)}  \qquad & \text{ and } \qquad G_{ij} \coloneqq \left\{ t\in O^+_{ij} : \frac{\overline{v}_j(t)}{\overline{v}_i(t)}\leq \frac{\overline{v}_j(g^*_{ij})}{\overline{v}_i(g^*_{ij})}  \right\} \label{eq:defnGij}
\end{align}
 and otherwise, set $G_{ij}=\emptyset$. Similarly, if $I^* \cap O^-_{ij} \neq \emptyset$, write 
 \begin{align}
c^*_{ij}  \in \ \argmax_{c \in I^*_i \cap O^{-}_{ij}} \frac{|\overline{v}_i(c)|}{|\overline{v}_j(c)|} \qquad & \text{ and } \qquad B_{ij} \coloneqq \left\{ t \in O^{-}_{ij} :   \frac{|\overline{v}_i(t)|}{|\overline{v}_j(t)|} \leq \frac{|\overline{v}_i(c^*_{ij})|}{|\overline{v}_j(c^*_{ij})|} \right\} \label{eq:defnBij}
\end{align}
 and otherwise, set $B_{ij}=\emptyset$.  
We now define, for each $i \neq j$, (forbidden) set $F_{ij}\subseteq [m]$ as follows  
\begin{align}
    F_{ij} \coloneqq Q_{ij} \cup G_{ij} \cup B_{ij} \label{eq:defnFij_mix}
\end{align}
Note that the sets $Q_{ij}$, $G_{ij}$, and $B_{ij}$ are pairwise disjoint. 

\begin{restatable}{lemma}{cor-intersection-mixed}
\label{cor:intersection}
For each agent $i \in [n]$, we have $I_i^* = \bigcap_{j \neq i} F_{ij}$.
\end{restatable}
\begin{proof} We establish the statement by proving both directions of the inclusion.  
For the lemma we establish two claims, each establishing one direction of the inclusion, respectively. 

\begin{claim}\label{claim:inclusion} $I_i^* \subseteq \bigcap_{j \neq i} F_{ij}$
\end{claim}
\begin{proof}
Fix any item $s\in I^*_i$, we will show that $s\in \bigcap_{j \neq i} F_{ij}$. We separate the proof into two cases depending on whether $i$ considers $s$ as a good (i.e.,  $\overline{v}_i(s) > 0$) or a chore (i.e., $\overline{v}_i(s) < 0$). Recall that since the valuations $\overline{v}_i$s are non-degenerate, we have $\overline{v}_i (t)\neq 0$ for all $t\in [m]$; thus, these two cases are exhaustive.  \\

\noindent
{\it Case {\rm I}: $\overline{v}_i(s) > 0$.} Here, agent $i$ considers item $s$ as a good.  Let $N'\subseteq [n] \setminus \{ i \}$ denote the set of agents who also value $s$ as a good. For each agent $j\in N'$, we have $s\in I^*_i \cap O^+_{ij}$. The definition of $g^*_{ij}\in I^* \cap O^{+}_{ij}$gives us  
$$
\frac{\overline{v}_j(s)}{\overline{v}_i(s)}\leq \frac{\overline{v}_j(g^*_{ij})}{\overline{v}_i(g^*_{ij})},
$$
and hence $s\in G_{ij}\subseteq F_{ij}$ (see equation (\ref{eq:defnGij})) for each $j \in N'$. Next, consider agents $k\in [n]\setminus ( N' \cup \{ i \})$ -- each such agent $k $considers item  $s$ as a chore. Since agent $i$ values $s$ as a good, we have that $s\in Q_{ik} \subseteq F_{ik}$, for each $k \in [n]\setminus ( N' \cup \{ i \})$. Combining these two observations, we conclude that  for every agent $j\neq i$, we have $s\in F_{ij}$ and, hence, $s\in \bigcap_{j\neq i} F_{ij} $.  \\

\noindent
{\it Case {\rm II}: $\overline{v}_i(s) < 0$.} In this case, given that $s\in I^*_i$, we have that $0> (w^*_i + \eta) \overline{v}_i(s) > (w^*_i + \eta) \overline{v}_j(s)$ for all $j\neq i$. This implies that every other agent also considers $s$ as a chore. 
Hence, for each $j\neq i$, we have $s\in I^*_i \cap O^{-}_{ij} $. By definition of $c^*_{ij} \in I^*_i \cap O^{-}_{ij}$, we have 
$$
\frac{|\overline{v}_i(s)|}{|\overline{v}_j(s)|} \leq \frac{ |\overline{v}_i(c^*_{ij})|}{|\overline{v}_j(c^*_{ij})| }  
$$
Hence, $s\in B_{ij} \subseteq F_{ij}$ for every $j\neq i$, which implies $s \in \bigcap_{j \neq i} F_{ij}$. 

Therefore, in both cases---whether agent $i$ considers  $s$ as a good or a chore---the item  $s \in \bigcap_{j \neq i} F_{ij}$. This establishes the claimed containment. 
\end{proof}

\begin{claim}\label{claim:revInclusion} 
$ \bigcap_{j \neq i} F_{ij} \subseteq I_i^*$
\end{claim}

\begin{proof}
 Fix any item $s\in \bigcap_{j \neq i} F_{ij}$, we will show that $s\in I_i^*$. As in the proof of \Cref{claim:inclusion},  we separate the analysis based on whether $i$ considers $s$ as a good or a chore. \\

 \noindent
{\it Case {\rm I}: $\overline{v}_i(s) > 0$.} Since $s\in \bigcap_{j \neq i} F_{ij}$,  and by definition $Q_{ij}$ and $G_{ij}$ are disjoint, it follows that for every agent $j\neq i$, we have  either $s\in Q_{ij}$ or $s\in G_{ij}$.

If $s\in Q_{ij}$, then agent $j$ considers $s$ as a chore. Further, the fact that agent $i$ values $s$ as a good implies
\begin{align*}
 (w^*_i+\eta) \overline{v}_i (s)> 0> (w^*_j+\eta) \overline{v}_j(s),     
\end{align*}
Therefore, $(w^*_i+\eta) \overline{v}_i (s)> (w^*_j+\eta) \overline{v}_j (s) $, for $s\in Q_{ij}$. 

Otherwise, if $s\in G_{ij} $, then there exists an item $g^*_{ij}\in I^*_i \cap O^+_{ij}$ that, by definition,  satisfies 
\begin{align}
    \frac{\overline{v}_j (s)}{\overline{v}_i (s)} \leq \frac{\overline{v}_j (g^*_{ij})}{\overline{v}_i (g^*_{ij})}.\label{eq:IneqGoods}
\end{align}
Since $g^*_{ij}\in I^*_i$, we know that $ (w^*_i + \eta)\overline{v}_i (g^*_{ij})> (w^*_j + \eta)\overline{v}_j (g^*_{ij})$. Given that both values  $\overline{v}_i (g^*_{ij}), \overline{v}_j (g^*_{ij})$ are positive, we obtain 
\begin{align*}
    \frac{\overline{v}_j (g^*_{ij}) }{\overline{v}_i (g^*_{ij})}< \frac{(w^*_i + \eta)}{(w^*_j + \eta)}.
\end{align*}
Combining with \Cref{eq:IneqGoods}, we get that $(w^*_i+\eta) \overline{v}_i (s)> (w^*_j+\eta) \overline{v}_j (s) $ whenever $s\in G_{ij}$. 

It follows that for any agent $j\neq i$, we have either $s\in Q_{ij}$ or $s\in G_{ij}$, and in both cases $(w^*_i+\eta) \overline{v}_i (s)> (w^*_j+\eta) \overline{v}_j (s)$ is satisfied. Since this strict inequality holds for each $j\neq i$, the welfare maximizing condition for $I^*_i$
gives us $s\in I^*_i$. \\

\noindent
{\it Case {\rm II}: $\overline{v}_i(s) < 0$.} Note that, if there exists an agent $k$ for whom $s$ is a good, then we have that $s\notin F_{ik}$. This non-containment contradicts the assumption that $s\in \bigcap_{j\neq i}F_{ij}$. Hence, for all $j\neq i$, agent $j$ values $s$ as a chore. Given that $s\in F_{ij}$ and both agents $i$ and $j$ value $s$ as a chore, it follows that $s\in B_{ij} $, for each $j\neq i$.

Further, the fact that $s\in B_{ij}$ implies that there  exists an item $c^*_{ij} \in I^*_{i} \cap O^{-}_{ij}$ that satisfies 
\begin{align}\label{eq:chore1}
    \frac{|\overline{v}_i(s)|}{|\overline{v}_j(s)|} \leq \frac{ |\overline{v}_i(c^*_{ij})|}{|\overline{v}_j(c^*_{ij})| }.  
\end{align}
Since $c^*_{ij} \in I^*_{i}$, we have $ (w_i + \eta)\overline{v}_i (c^*_{ij})> (w_j + \eta)\overline{v}_j (c^*_{ij})$. Also, given that both values $\overline{v}_i (c^*_{ij}), \overline{v}_j (c^*_{ij})$ are negative, the last equation reduces to 
\begin{align*}
    \frac{ |\overline{v}_i (c^*_{ij}) |}{|\overline{v}_i (c^*_{ij})|}< \frac{(w_j + \eta)}{(w_i + \eta)}.
\end{align*}
This strict inequality along with \Cref{eq:chore1}, and the fact that both agents $i$ and $j$ have negative value for $s$, give us $(w_i+\eta) \overline{v}_i(s)>(w_j+\eta) \overline{v}_j(s)$ for each $j\neq i$. That is, considering weighted welfare maximization, we obtain the desired containment: $s\in I^*_i$.

Overall, in both Case {\rm I} and {\rm II}, item $s\in I^*_i$, and the desired inclusion stands proved.  
\end{proof}

For each agent $i\in [n]$, Claims \ref{claim:inclusion} and \ref{claim:revInclusion} together imply $I_i^* = \bigcap_{j \neq i} F_{ij}$. The lemma stands proved.
\end{proof}

\subsection{Proof of Theorem \ref{theorem:fixed-agents}}
Recall that $\calA^*$ denotes an $\EFR{(n-1)}$ and $\PO$ allocation that maximizes weighted social welfare under the weight vector $w^* \in \Delta_{n-1}$ and the non-degenerate values $\overline{v}_i$. Further, we have the associated reallocation subset $R^*\subseteq [m]$ of size at most $(n-1)$. 

For the sake of analysis, we note that if we are given the $O(n^2)$ separating items $\{g^*_{ij}, c^*_{ij}\}$ for all $i \neq j$, along with the following sets, then we can use this partial information about $\calA^*$ to in fact  construct this allocation in its entirety: 
\begin{itemize}
\item $R^* \subseteq [m]$ of size at most $(n-1)$. There are $O(m^{n-1})$ many possible choices for $R^*$. 
\item For each item $t \in R^*$, the set $ D(t,w^*) \subseteq [n]$. There are $O(2^{n(n-1)})$ many choices for these sets, for each choice of $R^*$. 
\item For each agent $i \in [n]$, either $I^*_i = \emptyset$, or the separating items $\{g^*_{ij}, c^*_{ij}, \}_{j \neq i} $. There are $O(n^2)$ separating items and, hence, $ m^{O(n^2)}$ many possible choices for them. 
\end{itemize}
Indeed, given  $\{g^*_{ij}, c^*_{ij} \}_{j \neq i} $ we can we first compute the sets $F_{ij}$s (see equations (\ref{eq:defnGij}), (\ref{eq:defnBij}), and (\ref{eq:defnFij_mix})) and then $I^*_i = \cap_{j \neq i} F_{ij}$ (see Lemma \ref{cor:intersection}) for each agent $i \in [n]$.
The remaining items $R^* = [m] \setminus \left( \cup_i I^*_i \right)$ can be assigned following $D(t,w^*)$s for each $t \in R^*$. This gives us the entire allocation $\calA^*$. 

We iterate over all possible choices of the above-mentioned sets -- this overall entails an $O(m^{n^5})$-time exhaustive search. Write $R$, $D(t)$, and $\{g_{ij}, c_{ij} \}_{j \neq i}$ to denote an enumerated choice of the listed sets. 

Using the `guessed' separating items $\{g_{ij}, c_{ij} \}_{j \neq i}$, we first compute the sets $F_{ij}$s (following equation (\ref{eq:defnFij_mix})) and candidate $I_i = \cap_{j \neq i} F_{ij}$ for each agent $i \in [n]$. Then, we check whether the computed $I_i$s, along with the guessed $R$, partition $[m]$, or not. If a partition is indeed induced, then we assign the items $t \in R$ following $D(t)$s to obtain a candidate allocation $\calA$. Finally, we check whether the candidate $\calA$ is $\EFR{(n-1)}$ and $\PO$: \\

\noindent   
{\it Checking for $\EFR{(n-1)}$.} Given that we have candidate $I_i$s and $R$, we enumerate all possible $\mathcal{O}(n^{n-1})$ assignments of $R$ (respecting $D(t)$s) and test whether, for each agent $i \in [n]$, there exists an envy-free allocation in this enumerated list. \\  

\noindent
{\it Checking for $\PO$.} Given \Cref{lem:POwrtInput}, it suffices to find a weight vector $w$ such that the candidate allocation $\calA$ satisfies $\calA \in \overline{\Opt}_\eta(w)$. Towards this, consider the following system of linear inequalities with decision variables $w_1, w_2, \ldots, w_n \geq 0$.
 \begin{align*}
 (w_i+ \eta) \overline{v}_i(t) & \geq  (w_j+ \eta) \overline{v}_j(t)  \qquad \text{for each } t \in I_i \text{ and each } j \neq i. \\
 (w_i+ \eta) \overline{v}_i(t) & =  (w_j + \eta) \overline{v}_j(t) \qquad \text{for each } t \in R \text{ and each }  i, j \in D(t).
 \end{align*}
We can test the feasibility of this linear program in polynomial time and, hence, determine whether there exists a vector $w \in \Delta_{n-1}$ such that $\calA \in \overline{\Opt}_\eta(w)$. 

Since an $\EFR{(n-1)}$ and $\PO$ allocation $\calA^*$ is guaranteed to exist (\Cref{thm:EFRPO}), one of the enumerated choices will yield a linear program for which the underlying $w^*$ is a feasible solution.  
Further, any allocation $\calA$ that clears the above-mentioned tests is guaranteed to be $\EFR{(n-1)}$ and $\PO$. 

Overall, these soundness and completeness arguments ensure that an $\EFR{(n-1)}$ and $\PO$ allocation can be computed in $m^{\mathrm{poly}(n)}$ time. This completes the proof of the theorem. 
\section{Existence and Computation of $\mathrm{EFR}$ Allocations}
\label{sec:existence-efrk}
Having established results for $\mathrm{EFR}$ and $\PO$ in conjunction, in this section we address $\mathrm{EFR}$ by itself. We show that, for indivisible mixed manna with additive valuations, an $\EFR{(n-1)}$ allocation can be efficiently computed (Theorem \ref{thm:ERFexistence}).\footnote{This existence of such an allocation is implied by Theorem \ref{thm:EFRPO}.} When all the items are chores, the $(n-1)$ bound here is tight: there exist chore-division instances that do not admit an $\EFR{(n-2)}$ allocation (Theorem \ref{thm:EFRnon-existence}). 

Notably, we can achieve improved bounds when all the items are goods: For fair division of indivisible goods under additive valuations, an $\EFR{(\lfloor n/2 \rfloor)}$ allocation always exists can can be computed in polynomial time (\Cref{thm:EFRgoods}).

\subsection{$\EFR{(n-1)}$ Allocation of Mixed Manna}

\begin{restatable}{theorem}{ERFexistence}
    \label{thm:ERFexistence} For any fair division instance $\langle [n], [m], \{v_i\}_{i \in [n]} \rangle$ with mixed manna and additive valuations, there always exists an allocation that is both {\(\EFR{(n-1)}\)} and $\EF{1}$. Further, such an allocation can be computed in polynomial time.
\end{restatable}

\begin{proof} Let $\calA$ be an $\EF{1}$ allocation in the given mixed-manna instance. Such a fair allocation is guaranteed to exist and can be computed efficiently \cite{azizetalFairallocation22} and \cite{BSV21}. Furthermore, we can assume that, under $\calA$, one of the agents $\widehat{i} \in [n]$ is envy-free. This follows from the fact that one can resolve the top-trading envy cycles in the envy graph of $\calA$ and, hence, obtain a sink (i.e., an envy-free) node; see Lemma 6 in \cite{BSV21} for details regarding resolution of top-trading envy cycles.  

Hence, $\widehat{i}$ is envy-free under $\calA=(A_1,\ldots, A_n)$ and the allocation is $\EF{1}$ for all the $(n-1)$ remaining agents. For each remaining agent $i \in [n] \setminus \left\{ \widehat{i} \right\}$, we will identify an item $t_i \in [m]$ and show that, with set $R \coloneqq \{ t_i \}_{i \in [n] \setminus \left\{\widehat{i} \right\}}$, allocation $\calA$ upholds the $\EFR{(n-1)}$ criterion.  

In particular, for each $i \in [n] \setminus \left\{ \widehat{i} \right\}$, write $c_i$ to denote the lowest-valued chore in $A_i$ and $g_i$ to denote the highest-valued good in $[m] \setminus A_i$, i.e., 
\begin{align}
    c_i \in \argmin_{t \in A_i: \ v_i(t) < 0} \ v_i(t) \qquad \text{and} \qquad g_i \in \argmax_{t \in [m] \setminus A_i: \ v_i(t) \geq 0} \ v_i(t) 
\end{align}
Write item $t_i \in \argmax_{t \in \{c_i, g_i \}} \ |v_i(t)|$ and set of items $R \coloneqq \{ t_i \}_{i \in [n] \setminus \left\{\widehat{i} \right\}}$. We will complete the proof next by showing that, for each $i \neq \widehat{i}$, reassigning item $t_i \in R$ in $\calA$ leads to an envy-free allocation. 

There are two cases here: {(i)} Either $t_i = c_i$, (ii) or $t_i = g_i$. 

In case (i) (when $t_i = c_i)$, we obtain an allocation $\calA^i=(A^i_1,\ldots, A^i_n)$ that is envy-free for $i$ by assigning $t_i$ to any agent besides $i$; i.e., $A^i_i= A_i \setminus \{t_i \}$. Recall that allocation $\calA=(A_1, \ldots, A_n)$ is $\EF{1}$ for $i$. In particular, against any agent $j$, if the $\EF{1}$ criterion for $i$ holds by removal of a chore $c \in A_i$, then we have \[ v_i(A^i_i) = v_i(A_i) - v_i(t_i) = v_i(A_i) - v_i(c_i) \geq v_i(A_i) - v_i(c) \geq v_i(A_j) \geq v_i (A^i_j).\]
Hence, we have envy-freeness for $i$ against the considered agent $j$. 
Otherwise, if the $\EF{1}$ criterion for $i$ holds by removal of a good $g \in A_j$, then again 
\[ v_i(A^i_i) = v_i(A_i) - v_i(t_i) \geq v_i(A_i) + v_i(g_i) \geq v_i(A_i) + v_i(g) \geq v_i(A_j) \geq v_i (A^i_j).\]
Here, we use the facts that $v_i(t_i) = v_i(c_i) <0$ and $|v_i(t_i)| \geq v_i(g_i)$; these inequalities follow from the definition of the items $c_i$, $g_i$, and $t_i$. Hence, in case (i), agent $i$ is envy-free under allocation $\calA^i$.

In case (ii) (when $t_i = g_i)$, we obtain the envy-free allocation $\calA^i=(A^i_1, \ldots, A^i_n)$ by allocating $t_i = g_i$ to agent $i$, i.e., $A^i_i = A_i \cup \{ g_i \}$. Now, if, in $\calA$, $\EF{1}$ holds for $i$ by removal of a good $g \in A_j$, for any agent $j \in [n]$, then
\[v_i(A^i_i) = v_i(A_i) + v_i(g_i) \geq v_i(A_i) + v_i(g) \geq v_i(A_j) \geq v_i (A^i_j). \]
Alternatively, if $\EF{1}$ against $j$ holds by the removal of a chore $c \in A_i$, then 
\[v_i(A^i_i) = v_i(A_i) + v_i(g_i) \geq v_i(A_i) - v_i(c_i) \geq v_i(A_i) - v_i(c) \geq v_i(A_j) \geq v_i (A^i_j).\]
Therefore, in case (ii) as well, we have an envy-free allocation $\calA^i$ for $i$.

Overall, with $|R| \leq (n-1)$, we obtain that the allocation $\A$ is \(\EFR{(n-1)}\). The theorem stands proved.
\end{proof}

We next show that the guarantee obtained in \Cref{thm:ERFexistence} is tight.    

\begin{theorem} 
\label{thm:EFRnon-existence}
Given a fair division instance $\langle [n], [m], \{v_i\}_{i \in [n]} \rangle$ with mixed manna and additive valuations, an $\EFR{k}$ allocation with $k < (n-1)$ is not guaranteed to exist.
\end{theorem}

\begin{proof} 
    Consider an instance with \(n\) agents and \(m=(n-1)\) identical chores, each of which valued at \(-1\), i.e., \(v_i(t) = -1\) for all agents $i \in [n]$ and chores $t \in [n-1]$. 
    
    Assume, towards contradiction, that this instance admits an \(\EFR{(n-2)}\) allocation $\A=(A_1, \ldots, A_n)$ with the set $R$ of size at most $(n-2)$. Since $|R| \leq n-2$, there exists an agent $a$ with a chore $c \in A_{a}$ that is not in the set $R$, i.e., $c \in A_a \setminus R$. Write $\calA^a=(\calA^a_1, \ldots, \calA^a_n)$ to denote the allocation that is envy-free for $a$ and obtained by reassigning chores from within $R$. 

    Since $c \in A_a \setminus R$, it must be the case that chore $c$ continues to be with agent $a$ in $\calA^a$, i.e., $c \in \calA^a_a$. At the same time, given that the number of chores in the instance is $(n-1)$, there exists an agent $b \in [n]$ with an empty bundle in $\calA^a$. This, however, contradicts the envy-freeness of $\calA^a$ for agent $a$; we have $v_a(A^a_b) = 0$ and $v_a(A^a_a) \leq -1$. 
    
    Therefore, by way of contradiction, we obtain that the instance does not admit an \(\EFR{(n-2)}\) allocation.
\end{proof}

\subsection{$\EFR{\lfloor \frac{n}{2} \rfloor}$ Allocation of Goods}
This section establishes that, when all the items are goods (i.e., have nonnegative values), then, among $n$ agents with additive valuations, an $\EFR{\lfloor \frac{n}{2} \rfloor}$ allocation always exists and can be computed in polynomial time. This existential guarantee for goods is tight -- there exist instances in which no allocation of the goods is $\EFR{(\lfloor \frac{n}{2} \rfloor-1)}$. Also, recall that in the case of chores, even an $\EFR{(n - 2)}$ allocation might not exist (\Cref{thm:EFRnon-existence}). Together, these results highlight an interesting dichotomy between indivisible goods and chores.

Our algorithm for finding $\EFR{\lfloor \frac{n}{2} \rfloor}$ allocations (Algorithm~\ref{Alg:EFRgoods}) builds upon the round-robin method with a judiciously selected `picking sequence' over the agents. The algorithm maintains two complementary subsets of agents: the {active} agents $V \subseteq [n]$ and the {deferred} agents $D \subseteq [n]$. In addition, the algorithm updates two subsets of goods: the goods $R \subseteq [m]$, to be reallocated eventually, and the set of unallocated goods $G \subseteq [m]$. Initially, all agents are active and all the goods are unallocated, i.e., $V = [n]$, $D = \emptyset$, $R = \emptyset$, and $G = [m]$.

The algorithm populates a partial allocation of the goods as it iterates. Specifically, in each iteration (of the outer while-loop of Algorithm \ref{Alg:EFRgoods}) every agent receives a single new good. Each iteration of the algorithm begins with a {conflict resolution phase}. Every active agent $i \in V$ identifies her most-valued goods, $M_i$, among the unallocated ones, i.e., $M_i = \arg\max_{g \in G} \ v_i(g)$. Considering $M_i$s for the active agents $i \in V$, the algorithm defines, for each unallocated good $g \in G$, the set of conflicting (active) agents $N_g \coloneqq \{i \in V : g \in M_i \}$. 

The conflict phase (inner while-loop of the algorithm) continues while there exists a good $g \in G$ with $|N_g| \geq 2$. That is, the phase executes as long as there exists an unallocated good $g \in G$ that appears in the most-preferred set of two or more active agents: $M_i \cap M_j \neq \emptyset$ for any two distinct $i,j \in V$. We select a good $\widetilde{g} \in G$ with the maximum number of conflicts and move it into the set $R$. In addition, the algorithm transfers all the agents in $N_{\widetilde{g}}$ to the deferred set $D$. The active set $V$ is updated accordingly. Note that this phase (inner while-loop) continues until agents in $V$ have distinct most-valued goods in $G$. 

The algorithm then executes the {picking phase} (Line \ref{line:pick-ph}) wherein, in lexicographic order, each active agent $i \in V$ selects a most-valued good from $G$. Subsequently, the remaining (deferred) agents $j \in D$ do the same. The algorithm repeats the two phases until $G \neq \emptyset$. Finally, we assign the entire set of goods $R$---kept aside until now---to an arbitrary agent, say agent $1$ and return the computed allocation $\calA$. 

Lemma \ref{lem:invariants} shows that the algorithm maintains---as an invariant---envy-freeness for the active agents and $\EF{1}$ for the deferred agents. In the analysis below (proof of Theorem \ref{theorem:EFRgoods}), we will show that $|R| \leq n/2$. This bound will follow from the observation that the addition of each good in $R$ causes the inclusion of at least two agents in the deferred set $D$. This set $R$ will certify that the returned allocation is $\EFR{\lfloor \frac{n}{2}\rfloor}$.

\begin{algorithm}

\KwIn{Fair division instance \(\langle [n], [m], \{v_i\}_{i \in [n]} \rangle\) with only goods.}
\KwOut{An \(\EFR{\lfloor n/2 \rfloor}\) allocation $\calA$.}

\BlankLine
Initialize set of agents \(V = [n]\) and \(D = \emptyset\) along with set of goods \(R = \emptyset\) and \(G = [m]\).

Also, initialize bundles $A_i = \emptyset$ for the agents $i \in [n]$.
\BlankLine

\While{\(G \neq \emptyset\)}{
For each agent $i \in V$, define the set of most-preferred unallocated goods $M_i \coloneqq \arg\max_{g \in G} v_i(g)$ and, for each good $g \in G$, define $N_g \coloneqq \{ i \in V : g \in M_i \}$. 
    \BlankLine  
    \tcp{Conflict Resolution Phase}
     \While{there exists $g \in G$ with $|N_g| > 1$}{
    Select a good with most conflicts, $\widetilde{g} \in \arg\max_{g \in G} |N_g|$, and  
         update \\ \(G \gets G \setminus \{ \widetilde{g}\}\) and \(R \gets R \cup \{ \widetilde{g} \}\) along with \  $D \gets D \cup N_{\widetilde{g}}$ and \(V \gets V \setminus N_{\widetilde{g}}\). \label{line:conflict} \\
        
        \BlankLine
        Also, update $M_i \coloneqq \arg\max_{g \in G} v_i(g)$, for each $i \in V$, and $N_g \coloneqq \{ i \in V : g \in M_i \}$, for each $g \in G$.
     }

    \tcp{Picking Phase}
    For each active agent $i \in V$, select good $\widehat{g}_i \in \arg\max_{g \in G} v_i(g)$ and update $A_i \gets A_i \cup \{ \widehat{g}_i \}$  along with \(G \gets G \setminus \{ \widehat{g}_i\}\).  For each deferred agent $j \in D$, select good $\widehat{g}_j \in \arg\max_{g \in G} v_j(g)$ and update $A_j \gets A_j \cup \{ \widehat{g}_j \}$  along with \(G \gets G \setminus \{ \widehat{g}_j\}\). \label{line:pick-ph} \\
}

Set $A_1 \gets A_1 \cup \{R \}$ and return allocation $\calA = (A_1, \ldots, A_n)$.
\caption{Conflict-Aware Picking Sequence}\label{Alg:EFRgoods}
\end{algorithm}

\begin{restatable}{lemma}{lem_invariants}
\label{lem:invariants}
In each iteration of the outer while-loop of \Cref{Alg:EFRgoods} and under the maintained allocation $\calA$ it holds that  
\begin{enumerate}[label=(\roman*)]
\item Every active agent $i \in V$ is envy-free.
\item Every deferred agent $j \in D$ is $\EF{1}$ and becomes envy-free if allocated all the goods in $R$. 
\end{enumerate}
\end{restatable}
\begin{proof}
We prove lemma via induction over the number of iterations of the outer while-loop of the algorithm. 

\noindent {\it Base Case:} Initially, all agents are active, $V = [n]$, no goods are allocated or reserved, $A_i = \emptyset$ for each agent $i$. Hence, the maintained (partial) allocation $\calA$ is envy-free for all the agents.

\noindent {\it Induction Step:} Assume that the invariants hold at the beginning of the $k$th iteration of the outer while-loop. We will show that the they also hold at the end of the iteration and, hence, establish the inductive step. 

For any iteration count $\ell \in \mathbb{Z}$, write $\mathcal{A}^\ell=(A^\ell_1, \ldots ,A^\ell_n)$ to denote the (partial) allocation maintained by the algorithm at end of its $\ell$th iteration. Also, write $\widehat{G}_\ell$ to denote the set of goods assigned in the $\ell$th iteration (see Line \ref{line:pick-ph}). In particular, $\calA^{k-1}$ is the allocation at the beginning of the $k$th iteration, and $\calA^k$ is the one at the end. Further, the goods in $\widehat{G}_k$ are assigned in $\calA^k$, but not in $\calA^{k-1}$. 

Given that the invariants (i) and (ii) hold under $\calA^{k-1}$, we will next show that they are, respectively, satisfied by $\calA^k$ as well.   

\smallskip
\noindent \emph{(i) Every active agent is envy-free.} 
By construction, an agent $i$ that is active at the end of the $k$th iteration must have been active in all previous iterations. Hence, by the induction hypothesis, we have that agent $i$ was envy-free under $\calA^{k-1}$. Further, the conflict resolution and the fact that $i$ remains active ensure that $i$ gets to select $\widehat{g}_{i,k}$ -- its most preferred good among unallocated ones in Line \ref{line:pick-ph} of iteration $k$; in particular, $v_i(\widehat{g}_{i,k}) \geq v_i(g')$ for all $g' \in \widehat{G}_k$. 

Consider any other agent $a \in [n]$ and let $\widehat{g}_{a,k} \in \widehat{G}_k$ be the item assigned to $a$ in the $k$th iteration.\footnote{The argument directly holds if $a$ receives no item in the iteration, i.e., if $G$ becomes empty in between the picking phase.} Given that $v_i(\widehat{g}_{i,k}) \geq v_i(\widehat{g}_{a,k})$ and $v_i(A^{k-1}_i) \geq v_i(A^{k-1}_a)$, we obtain the envy-freeness for $i$: 
\begin{align*}
    v_i(A^k_i) = v_i(A^{k-1}_i) + v_i(\widehat{g}_{i,k}) \geq v_i(A^{k-1}_a) + v_i(\widehat{g}_{a,k}) = v_i(A^k_a). 
\end{align*}
Inductively, this establishes the envy-freeness of active agents.

\smallskip
\noindent \emph{(ii) Every deferred agent is $\EF{1}$, and becomes envy-free if allocated all of $R$.}
Consider any agent $j$ that is in the deferred set $D$ at the end of the $k$th iteration.  Also, fix any agent $a \in [n]$ -- we will next show that the invariant continues to hold for $j \in D$ against $a \in [n]$.

Here, for each agent $i \in [n]$ and iteration $\ell$, we write $\widehat{g}_{i, \ell}$ to denote the good selected by $i$ in the $\ell$th iteration (Line \ref{line:pick-ph}).  

Also, let $f \leq k$ denote the iteration in which agent $j$ was moved into $D$, and  let $\widetilde{g}_f$ denote the good that induced this movement (in Line \ref{line:conflict}). By construction, $\widetilde{g}_f \in R$. Also, note that, among all the goods that were unallocated at the beginning of the $f$th iteration, good $\widetilde{g}_f$ was a most-valued one by agent $j$. Hence, for the good selected by agent $a$ in the $f$th iteration, $\widehat{g}_{a,f}$, we have 
\begin{align}
    v_j\left( \widetilde{g}_f \right) \geq v_j \left( \widehat{g}_{a,f} \right) \label{ineq:faf}
\end{align}
In addition, note that agent $j$ was active before the $f$th iteration, and, hence, envy free: $v_j(A^{f-1}_j) \geq v_j(A^{f-1}_a)$. Therefore, $v_j(A^{f-1}_j) \geq v_j(A^{f-1}_a) = v_j(A^f_a \setminus \{ \widehat{g}_{a,f}\})$.

Moreover, in and after the $f$th iteration, agents $a$ and $j$ selected goods alternatively from among the unallocated ones. Hence, for iterations $\ell \in \{f, \ldots, k-1\}$, we have $v_j(\widehat{g}_{j, \ell}) \geq v_j( \widehat{g}_{a, \ell+1})$. These observations imply the $\EF{1}$ criterion holds for $j$ in iteration $k$:
\begin{align}
 v_j(A^k_j) & = v_j(A^{f-1}_j) + v_j\left(\{\widehat{g}_{j,f}, \ldots, \widehat{g}_{j,k}\}\right) \notag \\
 & \geq v_j(A^{f-1}_j) + v_j\left(\{\widehat{g}_{j,f}, \ldots, \widehat{g}_{j,k-1}\}\right) \notag \\ 
& \geq v_j(A^{f-1}_j) + v_j\left(\{\widehat{g}_{a,f+1}, \ldots, \widehat{g}_{a,k}\}\right) \notag \\ 
& \geq v_j \left(A^f_a \setminus \{ \widehat{g}_{a,f}\} \right) + v_j\left(\{\widehat{g}_{a,f+1}, \ldots, \widehat{g}_{a,k}\}\right) \notag \\ 
& = v_j \left(A^k_a \setminus \{ \widehat{g}_{a,f}\} \right) \label{ineq:ind-EFone}
\end{align}
Equations (\ref{ineq:ind-EFone}) and (\ref{ineq:faf}) along with the fact that $\widetilde{g}_f \in R$ establish the second part of invariant (ii) (i.e., deferred agent $j$ becomes envy-free if it receives all of $R$): $v_j(A^k_j \cup R) \geq v_j(A^k_a)$.

This completes the induction step and the lemma stands proved. 
\end{proof}

Using Lemma \ref{lem:invariants}, we next prove that \Cref{Alg:EFRgoods} computes an $\EFR{\lfloor \frac{n}{2} \rfloor}$ allocation in polynomial time.

\begin{theorem}\label{thm:EFRgoods}
\label{theorem:EFRgoods}
Given any fair division instance of goods with $n$ agents and additive valuations, \Cref{Alg:EFRgoods} computes an $\EFR{\lfloor \frac{n}{2} \rfloor}$ allocation in polynomial time.
\end{theorem}
\begin{proof}
Each iteration of the outer while-loop in \Cref{Alg:EFRgoods} runs in polynomial time, since it entails populating subsets (of agents and goods) and counting. Further, note that the number of iterations is at most $m$. Hence, the algorithm terminates in polynomial time.  

By Lemma~\ref{lem:invariants}, at the end of the algorithm, under the returned allocation $\calA$ every active agent is envy-free and every deferred agent becomes envy-free upon receiving all goods in the set $R$. Hence, we can achieve envy-freeness for each agent by reassigning goods from within $R$ in $\calA$. That is, $\calA$ is $\mathrm{EFR}$ with respect to the subset of items $R \subseteq [m]$. We will complete the proof of the theorem by showing next that $|R| \leq \lfloor \frac{n}{2} \rfloor$.

For each reserved good $\widetilde{g} \in R$, let $N_{\widetilde{g}}$ denote the set of agents who moved into the deferred $D$ when $\widetilde{g}$ was included in $R$; see Line \ref{line:conflict}. Note that $|N_{\widetilde{g}}| \ge 2$. Also, note that the sets $\{N_g\}_{g \in R}$ are pairwise disjoint -- an agent is deferred due to exactly one good. Hence, summing over all such disjoint sets gives us $2|R|\le \sum_{g \in R} |N_g| \le n,
$. Therefore, as required, $|R| \le \lfloor n/2 \rfloor$. 

Overall, we obtain that the allocation returned by \Cref{Alg:EFRgoods} is $\EFR{\lfloor n/2 \rfloor}$. This completes the proof of the theorem. 
\end{proof}

We note that the existential guarantee obtained here for goods is tight. In particular, there exist instances with all goods wherein no allocation is $\EFR{\left(\lfloor n/2 \rfloor - 1\right)}$.

\begin{theorem}
\label{thm:EFRgoodstight}Given a fair division instance $\langle [n], [m], \{v_i\}_{i \in [m]} \rangle$ with goods and additive valuations, an $\EFR{k}$ allocation with $k < \lfloor \frac{n}{2} \rfloor$ is not guaranteed to exist.
\end{theorem}
\begin{proof}
Consider an instance with an even number of agents $n$  and $m = n/2$ goods $\{g_1, \dots, g_{n/2} \}$. For each index $k \le n/2$, agents $(2k-1)$ and $2k$ both value good $g_k$ at $1$, and all other goods at $0$. In this instance, each good must belong to the re-allocatable set $R$. Indeed, if $g_k \notin R$, then in any allocation $\calA$ the envy of at least one of the agents $(2k-1)$ or $2k$ cannot be resolved even by the reassignment of all the remaining goods (which they value at 0). Hence, it must hold that $|R| \geq n/2$, i.e., the instance does not admit an $\EFR{\left(\lfloor n/2 \rfloor - 1\right)}$ allocation. 
\end{proof}

\paragraph{Remark:} Lemma~\ref{lem:invariants} and \Cref{thm:EFRgoods} show that in the partial allocation maintained at the termination of the outer while-loop all the goods in $[m] \setminus R$ are allocated and we have $\EF{1}$ among the agents. Write $\calA'$ to denote this partial allocation and recall that, for each agent $i$, the (complete) allocation $(A'_1, \ldots, A'_i \cup R, \ldots, A'_n)$ is envy-free for $i$.

To extend $\mathcal{A}'$ to a complete allocation (instead of assigning all of $R$ to a single agent) we can assign the goods in $R$ in a round-robin manner while respecting the  topological ordering of the envy graph $G_{\calA'}$; recall that $|R| \le \lfloor n/2 \rfloor$. This ensures that the resulting complete allocation is $\EF{1}$ for all the agents. In addition, the $\EFR{\lfloor n/2 \rfloor}$ guarantee continues to hold, since we have selected an allocation of $R$ starting with $\calA'$. This observation leads to the corollary below. 

\begin{corollary}
In fair division instances $\langle [n], [m], \{v_i\}_{i \in [n]} \rangle$ with all goods  and additive valuations, there always exists an allocation that is both $\EFR{\lfloor n/2 \rfloor}$ and $\EF{1}$. Further, such an allocation can be computed in polynomial time.
\end{corollary}

\section{Conclusion and Future Work}
This paper establishes that, when allocating mixed manna among $n$ agents with additive valuations, there always exist proximal, $\PO$ allocations $\calA^1, \ldots, \calA^n$ that are envy-free for the $n$ agents, respectively. Here, the proximity is quantified in terms of the symmetric differences between the allocation bundles; in particular, for any pair of these allocations $\calA^i=(A^i_1, \ldots, A^i_n)$ and $\calA^j=(A^j_1, \ldots, A^j_n)$, the union of the symmetric differences $\cup_{\ell \in [n]} \left( A^i_\ell \triangle A^j_\ell \right)$ is contained in a size-$(n-1)$ subset of items $R$. While the cardinality bound of $(n-1)$ is tight (specifically for chores), a relevant direction of future work is to strengthen the guarantee by bounding the size of the symmetric differences $\left( A^i_\ell \triangle A^j_\ell \right)$ individually. We note that establishing degree bounds in the bipartite graph $\calH$ considered in \Cref{claim:acyclic} would provide such a strengthening.  

This paper provides a novel application of the KKM Theorem in discrete fair division. Our approach complements the existing market-based techniques and contributes to the mixed manna frontier. Identifying further applications of this template in discrete fair division is another interesting direction.

\clearpage
\bibliographystyle{alpha}
\bibliography{references}

\appendix
\clearpage
\appendix
{\noindent\textbf{\huge Appendix}}
\bigskip

\section{Hardness of Deciding \(\EFR{k}\) Allocations}

While an $\EFR{(n-1)}$ allocation always exists and can be computed efficiently (\Cref{thm:ERFexistence}), we now show that deciding whether a given allocation $\A$ is \(\EFR{k}\) is NP-complete.

\begin{theorem}
    \label{thm:ERFhardness}
    Given a fair division instance $\langle [n], [m], \{v_i\}_{i \in [n]} \rangle$ with mixed manna, an allocation $\A$, and a positive integer \(k < n-2 \), deciding whether \(\A\) is \(\text{EFR-}k\) is NP-complete.
\end{theorem}
\begin{proof}
    We will, in fact, establish the hardness result for instances with identical valuations. 

    The problem at hand is in {\rm NP}: Given an allocation \(\A\), a set \(R\) of items of size at most \(k\), and allocations \(\A^i\) for each agent \(i\), we can verify in polynomial time that each \(\A^i\) is envy-free for agent \(i\) and that \(\A^i\) is obtained from \(\A\) by reallocating only items in \(R\).

    To establish the hardness of the $\mathrm{EFR}$ decision problem, we next provide a reduction from the {\rm NP}-complete \textsc{Partition} problem \cite{gareyjohnsonComputersintractability79}. In \textsc{Partition}, we are given a set \(S = \{s_1, s_2, \dots, s_k\}\) of positive integers that sum up to \(2T\), and we need to decide whether there exists a subset \(S_1 \subseteq S\) whose elements sum to exactly \(T\).

    Given an instance \(S\) of \textsc{Partition} with \(\sum_{i=1}^k s_i = 2T\), we construct a fair division instance $\langle [k+3], [2k+1], \{v_i\}_{i \in [k+3]} \rangle$ with \(2k+1\) chores, \(k+3\) agents, that have identical valuations $v_i = v$, and the input allocation \(\A\) as follows: 
     \begin{itemize}
    \item For the first $k$ chores $c_1, \ldots, c_k$ we set the value considering the $k$ elements in the given partition instance:  $v(c_i) = -s_i$ for all indices $1 \leq i \leq k$; recall that the agents' valuations are identical. 
    \item For the remaining $k+1$ chores $c_{k+1}, \ldots, c_{2k+1}$, we set  $v(c_\ell) = -T$, with index $k+1 \leq \ell \leq 2k+1$.
    \item The allocation $\A=(A_1, \ldots, A_{k+3})$ is defined as follows:
        \begin{itemize}
            \item Bundle $A_1 = \{c_1, \ldots, c_k\}$ and bundle $A_2 = \emptyset$,
            \item For the remaining agents $i \in \{3, \ldots, k+3\}$, set $A_i = \{c_{k-2+i}\}$ (i.e., each of the remaining $k+1$ agents gets one of the chores $c_{k+1}, \ldots, c_{2k+1}$).
        \end{itemize}
    \end{itemize}
    
    We claim that \(\A\) is \(\EFR{k}\) if and only if \(S\) admits an exact partition into two subsets of sum \(T\).
    \medskip\noindent\textbf{(If.)}\quad
    Suppose there is a partition \(S = S_1 \cup S_2\) with \(\sum_{s\in S_1}s = \sum_{s\in S_2}s = T\). Then, setting $R=A_1$, we have $|R| = |A_1| = k$. Next, we note that there exists an allocation $\B = (B_1,\ldots, B_{k+3})$ (defined below) that is envy-free for all the agents and can be obtained from $\calA$ by reallocating items from within $R$. Allocation $\B$ and the set $R$ certify that $\calA$ is \(\EFR{k}\).
    We obtain $\B$ as follows 
    \begin{itemize}
        \item $B_1 = \{c_i : s_i \in S_1\}$, i.e., agent 1 gets the chores corresponding to the items in $S_1$.
        \item $B_2 = \{c_i : s_i \in S_2\}$, i.e., agent 2 gets the chores corresponding to the items in $S_2$.
        \item $B_i = A_i$, for each $i \in \{3, \ldots, k+3\}$, i.e., each remaining agent gets the same bundle as in $\A$.
    \end{itemize}
    
    Each agent $i$ has a valuation of $-T$ for all the bundles in $\calB$, including its own. Hence, $\B$ is in fact envy-free for all the agents. In addition,  $\B$ differs from $\A$ only by reallocating the chores in $R$. Therefore, \(\A\) is \(\EFR{k}\).

    \medskip\noindent\textbf{(Only if.)}\quad
    Conversely, suppose that \(\A\) is an \(\EFR{k}\) allocation, with some $R \subset [2k+1]$ of size $|R|\le k$. 
    By the pigeonhole principle, since \(R\) contains at most \(k\) chores, there is at least one agent \(a \in \{3, \ldots, k+3\}\) such that the chore assigned to agent \(a\) in \(\A\) (i.e., in the bundle \(A_{a}\)) is not in \(R\). That is, \(A_{a} \cap R = \emptyset\). 

   Since \(\A\) is \(\EFR{k}\), there exists an allocation \(\A^a\) (obtained by reallocating items in \(R\)) such that agent \(a\) is envy-free in \(\A^a\). Next, note that since \(A_a \cap R = \emptyset\), agent \(a\)'s bundle \(A^a_a\) continues to contain the chore it received in $\calA$, i.e., \(v \left(A^a_a \right) \leq -T\).

    Further, the fact that agent \(a\) is envy-free in allocation \(\A^a\), it holds that $v(A^a_i) \leq -T$ for all agents $i \in [k+3]$. At the same time, the sum of values of all the chores satisfies $\sum_{i=1}^{k+3} v(A^a_i) = -2T - (k+1)T = -(k+3)T$. These observations imply that $v(A^a_i) = -T$ for all agents $i \in [k+3]$.   

    That is, under allocation $\calA^a$, the chores are partitioned in $(k+3)$ bundles $A^a_1, \ldots, A^a_{k+3}$ each with value $v(A^a_i) = -T$. Recall that the values of the chores are as follows:   
  \[
    \{-s_1, \ldots, -s_k\} \cup \underbrace{\{-T, \dots, -T\}}_{(k+1) \text{ copies}}
    \]
Considering these values and the partitioning of the chores induced by $\calA^a$, we obtain that the values $\{-s_1, \ldots, -s_k\}$ must admit a two-partition wherein each part sums exactly to $-T$. That is, the given integers \(\{s_1, \dots, s_k\}\) can be partitioned into two subsets of equal sum \(T\). Overall, the existence of an $\EFR{k}$ allocation implies the existence of a desired solution for the given \textsc{Partition} instance. This completes the reverse direction of the reduction. 
   
The theorem stands proved. 
\end{proof}

\section{Non-Implication of Envy-Freeness with Sharing}
\label{sec:EFRvsEFSharing}

The work of Sandomirskiy and Segal-Halevi \cite{sandomirskiysegalhaleviEfficientFair22} shows that, for mixed manna with additive valuations, there exist envy-free ($\EF{}$) and $\PO$ allocations in which at most $(n-1)$ items are fractionally assigned (and the remaining items are integrally allocated). 

We next show that such a sharing guarantee does not, in and of itself, provide a reasonable $\mathrm{EFR}$ bound; note that all allocations are $\EFR{m}$.

\begin{proposition}
For any $\varepsilon >0$, there exists an $\EF{}$ and $\PO$ allocation $\mathcal{X}$ in which $(n-1)$ items are fractionally assigned and it holds that none of the integral allocations obtained by rounding $\calX$ are $\EFR{\big( (1-\varepsilon)m \big)}$.       
\end{proposition}
\begin{proof}
Given $\varepsilon \in (0,1)$, we construct a fair division instance and allocation $\calX$ as follows. Set the number of agents as an odd integer $n \geq \frac{4}{\varepsilon}$ and the number of items $m \geq  n^2$. 

The first $(n-1)$ items are fractionally assigned among the first $(n-1)$ agents. Their values are set as follows: $ v_i(f) = 2$ for all agents $i \in [n-1]$ and items $1 \leq f \leq (n-1)$. For the first $(n-1)$ agents and the remaining items we have: $v_i(t) = -1$ for all $i \in [n-1]$ and $n \leq t \leq m$.    Next, we set the valuation, $v_n$, of the last agent: 
\begin{itemize}
\item $v_n(f) = -2$ for all items $f \in \{1,\ldots, \frac{n-1}{2}\}$.
\item $v_n(f') = 0$ for all items $f' \in \{\frac{n+1}{2}, \ldots, (n-1)\}$.
\item $v_n(t) = \frac{-1}{m-n+1}$ for all the remaining items $t \in \{n, \ldots,  m\}$.
\end{itemize}

To construct allocation $\calX$, we assign each of the first $(n-1)$ items uniformly among the first $(n-1)$ agents. That is, each agent $i \in [n-1]$ receives a $\frac{1}{n-1}$ fraction of each item $1 \leq f \leq (n-1)$. The remaining $(m-n+1)$ items are integrally assigned to the last agent $X_n = \{n, n+1, \ldots, m\}$.  

Note that the allocation $\calX$ is social-welfare maximizing and, hence, $\PO$. Further, it is envy-free: Each agent $i \in [n-1]$, values its own (fractional) bundle at $2$. Agent $i$'s value for the bundle of any other agent $j \leq (n-1)$ is also $2$ and $v_i(X_n) = -(m-n+1)$. Hence, envy-freeness holds for all agents $i \in [n-1]$. 

For agent $n$, the value $v_n(X_n) =   (m-n+1) \ \frac{-1}{m-n+1} = -1$. Recall that agent $n$ values $\frac{n-1}{2}$ of the fractionally assigned items at $-2$ and the remaining $\frac{n-1}{2}$ fractionally assigned items at $0$. Hence, agent $n$'s value for any other agent's bundle is equal to $\frac{n-1}{2} \times \left(-2 \right) \times \left(\frac{1}{n-1} \right) = -1$. That is, agent $n$ is envy-free under $\calX$ as well. 

Now, consider any integral allocation $\calB$ obtained by rounding $\calX$. For $\calB=(B_1, \ldots, B_n)$ we have $B_n \supseteq X_n$, since all the items in $X_n$ were already integrally assigned to agent $n$. We will next show that $\calB$ is not $\EFR{\big( (1-\varepsilon)m \big)}$.

Assume, towards the contradiction, that $\calB$ is $\EFR{\big( (1-\varepsilon)m \big)}$, i.e., there exists a subset $R \subset [m]$--of size $(1-\varepsilon)m$---such that reallocating items from $R$ yields an envy-free allocation $\calB^n=(B^n_1,\ldots, B^n_n)$ for agent $n$. Note that $B^n_n \supseteq B_n \setminus R$. 

The above-mentioned containments along with the bounds $|X_n| = (m-n+1)$ and $|R| = (1-\varepsilon) m$, give us that, even after reallocations, sufficiently many items from $X_n$ continue to be with agent $n$ (in $B^n_n$): 
\begin{align}
    |B^n_n \cap X_n| \geq \left| \left(B_n \setminus R \right) \cap X_n \right| \geq \left| \left(X_n \setminus R \right) \cap X_n \right| = |X_n \setminus R| \geq (m-n+1) - (1-\varepsilon) m \geq \varepsilon m - n \label{ineq:burnt-breadcrumbs}
\end{align}
Since $n \geq 4/\varepsilon$ and $m \geq n^2$, equation (\ref{ineq:burnt-breadcrumbs}) reduces to 
\begin{align}
    |B^n_n \cap X_n| & \geq \varepsilon m - n  \geq \frac{4}{n} m - n \geq \frac{3m}{n} \label{ineq:bread}
    \end{align}

Recall that agent $n$ has a non-positive value for each of the $m$ items; in particular, $v_n(t) = \frac{-1}{m-n+1}$ for each $t \in X_n$. Hence, from equation (\ref{ineq:bread}), we obtain 
\begin{align}
    v_n(B^n_n) \leq   \frac{3m}{n} \left(\frac{-1}{m-n+1}\right) \leq \frac{-3}{n} \label{ineq:breadtwo}
\end{align}
We will next show that, under $\calB^n = (B^n_1,\ldots,B^n)$ and for some agent $a$'s bundle it holds that $v_n(B^n_a) > \frac{-3}{n}$. This bound will contradict the envy-freeness of $\calB^n$ for agent $n$ and complete the proof. 

Recall that agent $n$ values the first $\frac{n-1}{2}$ items at $-2$, and for equally-many items agent $n$ has a value of $0$. Write $F$ to denote the first $\frac{n-1}{2}$ items, $F = \{ t \in [m] : v_n(t)= -2 \}$. Since $|F| = \frac{n-1}{2}$, at least half the agents $i \in [n-1]$ do not receive an item from $F$ in the allocation $\calB^n=(B^n_1, \ldots, B^n_n)$. That is, there exists subset of  agents $S \subseteq [n-1]$ such that $|S| \geq \frac{n-1}{2}$ and $B^n_i \cap F = \emptyset$ for each $i \in S$. Again, using the fact that agent $n$ has a non-positive value for all the items, we get 
\begin{align*}
\sum_{i \in S} v_n(B^n_i) \geq v_n \left( [m] \setminus F \right) = \left(\frac{n-1}{2} \right) \times 0 + (m-n+1) \left( \frac{-1}{m-n+1} \right) = -1. 
\end{align*}

Hence, there exists agent $a \in S$, with $v_n(B^n_a) \geq \frac{-1}{|S|} \geq \frac{-2}{n-1}$. Since $n \geq 4/\varepsilon > 3$, the above-mentioned inequalities contradict the envy-freeness of $\calB^n$ for agent $n$  
\[v_n(B^n_n) \leq \frac{-3}{n} < \frac{-2}{n-1} \leq v_n(B^n_a).\]

Therefore, by way of contradiction, we have that any integral allocation $\calB$, obtained by rounding $\calX$, cannot be $\EFR{\big( (1-\varepsilon)m \big)}$. This completes the proof of the proposition.   
\end{proof}

\end{document}